\documentclass[11pt,letterpaper]{article}

\usepackage{typearea}
\typearea{14}

\usepackage{comment,algorithm,algorithmic}


\usepackage{wrapfig}
\usepackage{color,colortbl}
\usepackage{float}
\usepackage{amsthm}
\usepackage{amsmath}
\usepackage{amssymb}
\usepackage{graphicx}
\usepackage{xspace}
\usepackage{nicefrac}
\usepackage[colorinlistoftodos,prependcaption,textsize=tiny]{todonotes}
\usepackage{thm-restate}

\usepackage{thmtools}
\usepackage{thm-restate}
\declaretheorem[name=Lemma]{lemma}

\definecolor{Darkblue}{rgb}{0,0,0.4}
\definecolor{Brown}{cmyk}{0,0.61,1.,0.60}
\definecolor{Purple}{cmyk}{0.45,0.86,0,0}

\usepackage[colorlinks,linkcolor=blue,filecolor=blue,citecolor=blue,urlcolor=blue,pdfstartview=FitH,pagebackref]{hyperref}
\usepackage[nameinlink]{cleveref}

\newtheorem{theorem}{Theorem}
\newtheorem{corollary}{Corollary}
\newtheorem{remark}{Remark}
\newtheorem{claim}{Claim}

\newtheorem{definition}{Definition}

\newcommand{\etal}{{\em et al.\ }}

\newcommand{\namedref}[2]{\hyperref[#2]{#1~\ref*{#2}}}

\newcommand{\conref}[1]{\hyperref[#1]{Condition~(\ref*{#1})}}

\newcommand{\supp}{{\rm supp}}
\newcommand{\E}{{\mathbb{E}}}

\newcommand{\R}{\mathbb{R}}
\newcommand{\Z}{\mathbb{Z}}

\newcommand{\poly}{{\rm poly}}

\newcommand{\tw}{{\rm tw}}
\newcommand{\Lip}{{\rm Lip}}

\newcommand{\eps}{\varepsilon}

\newcommand{\SPD}{\textsf{SPD}\xspace}
\newcommand{\PSPD}{\textsf{PSPD}\xspace}
\newcommand{\SPDs}{{\SPD}s\xspace}
\newcommand{\SPDdepth}{\textsf{SPDdepth}\xspace}

\def\cD{\ensuremath{\mathcal{D}}}

\def\cP{\ensuremath{\mathcal{P}}}

\def\inline#1:{\par\vskip 7pt\noindent{\bf #1:}\hskip 10pt}

\def\inline#1:{\par\vskip 7pt\noindent{\bf #1:}\hskip 10pt}

\def\blackslug{\hbox{\hskip 1pt \vrule width 4pt height 8pt
		depth 1.5pt \hskip 1pt}}

\def\QED{\quad\blackslug\lower 8.5pt\null\par}

\newcommand{\initOneLiners}{%
    \setlength{\itemsep}{0pt}
    \setlength{\parsep }{0pt}
    \setlength{\topsep }{0pt}
}

\newcommand{\alert}[1]{\textbf{\color{red}
		[[[#1]]]}\marginpar{\textbf{\color{red}**}}\typeout{ALERT:
		\the\inputlineno: #1}}
\definecolor{purple}{rgb}{0.294, 0, 0.71}

\floatstyle{ruled}
\newfloat{algorithm}{tbp}{loa}
\providecommand{\algorithmname}{Algorithm}
\floatname{algorithm}{\protect\algorithmname}

\usepackage{authblk}
\usepackage{pdfsync}

\begin{document}
	\author[]{Arnold Filtser \thanks{Work supported in part by ISF grant 1817/17, BSF Grant 2015813, Simons Foundation, and ISF grant 1042/22.}}
\affil[]{Bar-Ilan University,\hspace{10pt} Email: \texttt{arnold.filtser@biu.ac.il}}
\date{}
\begin{titlepage}
  \title{A face cover perspective to $\ell_1$ embeddings of planar graphs\thanks{A preliminary version of this paper was published in the proceedings of SODA 20 \cite{Fil20}.}}
  \maketitle
\begin{abstract}
It was conjectured by Gupta et al. [Combinatorica04] that every planar graph can be embedded into $\ell_1$ with constant distortion. However, given an $n$-vertex weighted planar graph, the best upper bound on the distortion is only  $O(\sqrt{\log n})$, by Rao [SoCG99].
In this paper we study the case where there is a set $K$ of terminals, and the goal is to embed only the terminals into $\ell_1$ with low distortion.
In a seminal paper, Okamura and Seymour [J.Comb.Theory81] showed that if all the terminals lie on a single face, they can be embedded isometrically into $\ell_1$.
The more general case, where the set of terminals can be covered by $\gamma$ faces, was studied by Lee and Sidiropoulos [STOC09] and Chekuri et al. [J.Comb.Theory13]. 
The state of the art is an upper bound of $O(\log \gamma)$ by Krauthgamer, Lee and Rika [SODA19].
Our contribution is a further improvement on the upper bound to $O(\sqrt{\log\gamma})$. 
Since every planar graph has at most $O(n)$ faces, any further improvement on this result, will be a major breakthrough, directly improving upon Rao's long standing upper bound.
Moreover, it is well known that the flow-cut gap equals to the distortion of the best embedding into $\ell_1$. Therefore, our result provides a polynomial time  $O(\sqrt{\log \gamma})$-approximation to the sparsest cut problem on planar graphs, for the case where all the demand pairs can be covered by $\gamma$ faces.
\end{abstract}

\vfill
\setcounter{tocdepth}{1}
\tableofcontents
\thispagestyle{empty}	
\end{titlepage}
\newpage
\section{Introduction}\label{sec:intro} 
Metric embeddings is a widely used algorithmic technique that have numerous applications, notably in approximation, online and distributed algorithms. In particular, embeddings into $\ell_1$ have implications to graph partitioning problems. Specifically, the ratio between the Sparsest Cut and the maximum multicommodity flow (also called flow cut gap) is upper bounded by the distortion of the optimal embedding into $\ell_1$ (see \cite{LLR95,GNRS04}).

Given a weighted graph $G=(V,E,w)$ with the shortest path metric $d_G$, and embedding $f: V \to
\ell_1$, the \emph{contraction} and
\emph{expansion} of $f$ are the smallest $\tau,\rho$, respectively, such that for every pair $u,v\in V$,
\[ \frac1\tau \cdot d_G(u,v)\leq \| f(u) - f(v) \|_1\leq
\rho\cdot d_G(u,v)~~. \] The \emph{distortion} of the embedding is $\tau \cdot \rho$. If $\tau=1$ (resp. $\rho=1$) we say that the embedding is non-contractive (expansive). If $\rho=O(1)$, we say that the embedding is Lipschitz.

In this paper we focus on embeddings of planar graphs into $\ell_1$. Rao \cite{Rao99} showed that every $n$-vertex planar graph can be embedded into $\ell_1$ with distortion $O(\sqrt{\log n})$. The best known lower bound is $2$ by Lee and Raghavendra \cite{LR10}.
A long standing conjecture by Gupta \etal \cite{GNRS04} states that every graph family excluding a fixed minor, and in particular planar graphs, can be embedded into $\ell_1$ with constant distortion.

\begin{wrapfigure}{r}{0.27\textwidth}
	\begin{center}
		\vspace{-20pt}
		\includegraphics[width=0.25\textwidth]{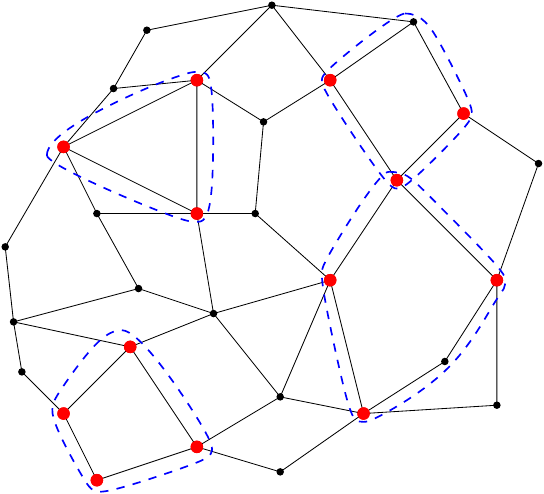}
		\vspace{-7pt}
	\end{center}
	\vspace{-10pt}
	\caption{\label{fig:facecover}\footnotesize The terminal vertices colored in red. The size of the face cover is $4$. The faces in the cover are encircled by a blue dashed lines.}
\end{wrapfigure}
Consider the case where there is a set $K\subseteq V$ of terminals, and we are only interested in embedding the terminals into $\ell_1$. 
This version is sufficient for the flow-cut gap equivalence, where the terminals are the vertices with demands. 
Better embeddings might be constructed when $K$ has a special structure. A \emph{face cover} of $G$ is a set of faces such that every terminal belongs to some face from the set (see \Cref{fig:facecover} for an illustration).
Given a drawing of $G$ in the plane, denote by $\gamma(G,K)$  the minimal size of a face cover.
It was shown by Hurkens, Schrijver and Tardos \cite{HST86}, that the result of Okamura and Seymour \cite{OS81} implies that if $\gamma(G,K)=1$, that is, if all the terminals lie on a single face, then $K$ embeds isometrically into $\ell_1$ (a special case is when $G$ is outerplanar).
For the general case, where $\gamma(G,K)=\gamma\ge 1$,
the methods of Lee and Sidiropoulos \cite{LS09} imply distortion of $2^{O(\gamma)}$.
Chekuri, Shepherd and Weibel \cite{CSW13} constructed an embedding with distortion of $3\gamma$.
Recently, Krauthgamer, Lee and Rika \cite{KLR19} managed to construct an embedding into $\ell_1$ with $O(\log \gamma)$ distortion by first applying a stochastic embedding into trees. This method has benefits, since trees are very simple to work with. Additionally, the result of \cite{KLR19} is tight w.r.t stochastic embedding into trees.
We improve upon \cite{KLR19} by embedding directly into $\ell_1$.
 \begin{restatable}[]{theorem}{MainEmbedding}
	\label{thm:main}
	Let $G=(V,E,w)$ be a weighted planar graph with a given drawing in the plane and $K\subseteq V$ a set of terminals. There is an embedding of $K$ into $\ell_1$ with distortion $O(\sqrt{\log \gamma(G,K)})$.
	Moreover, this embedding can be constructed in polynomial time.
\end{restatable}
Since every $n$-vertex graph has $O(n)$ faces, by setting $K=V$, \Cref{thm:main} re-proves the celebrated result of Rao \cite{Rao99}. Moreover, any improvement upon \Cref{thm:main} will be a major breakthrough. 

Using the Lipschitz extension problem (see \cite{LN05} and \Cref{thm:LipschitzExtension}) on planar graphs, we can extend the embedding of \Cref{thm:main} to the entire graph (see \Cref{sec:LipschitzExtension} for details).
 \begin{restatable}[]{corollary}{EmbeddingAllGraph}
	\label{cor:main}
	Let $G=(V,E,w)$ be a weighted planar graph with a given drawing in the plane and $K\subseteq V$ a set of terminals. There is an embedding $f:V\rightarrow \ell_1$ with expansion $O(\sqrt{\log \gamma(G,K)})$, and no contraction over $K$. Specifically $\forall u,v\in V$, $\|f(v)-f(u)\|_1\le O(\sqrt{\log \gamma(G,K)})\cdot d_G(u,v)$, and $\forall u,v\in K$, $\|f(v)-f(u)\|_1\ge d_G(u,v)$.
	Further, this embedding can be constructed in polynomial time.
\end{restatable}

\Cref{cor:main} has implication on the \emph{sparsest cut} problem.
Let $c:E\rightarrow \R_+$ be an assignment of capacities to the edges, and $d:{K\choose 2}\rightarrow \R_+$ assignment of demands to terminal pairs. The sparsity of a cut $S$ is the ratio between the capacity of the edges crossing the cut to the demands crossing the cut. The sparsest cut is the cut with minimal sparsity. \Cref{cor:main} implies:
\begin{corollary}\label{cor:sparset}
	Let $G=(V,E)$ be a weighted planar graph with a given drawing in the plane, $K\subseteq V$ a set of terminals, capacities $c:E\rightarrow \R_+$ and demands $d:{K\choose 2}\rightarrow \R_+$. Let $\gamma(G,K)=\gamma$.
	Then there is a polynomial time $O(\sqrt{\log \gamma})$-approximation algorithm for the sparsest-cut problem.
\end{corollary}
See \cite{LLR95,GNRS04,KLR19,CFW12} for further details.

\subsection{Technical Ideas}
Abraham \etal \cite{AFGN22}, among other results, constructed an $O(\sqrt{\log n})$-distortion embedding of planar graphs into $\ell_1$. 
This embedding is based on \emph{shortest path decompositions} (\SPD).
Even though the distortion is similar, the new embedding is very different from the classic embedding of Rao \cite{Rao99}.
An \SPD is a hierarchical decomposition of a graph using shortest paths. The first level of the partition is simply $V$. In level $i$, all the clusters are connected. To construct level $i+1$, we remove a single shortest path from every cluster of level $i$. Level $i+1$ consists of the remaining connected components. This process is repeated until all the vertices are removed. The \SPDdepth is the depth of the hierarchy.
Using cycle separators \cite{Mil86} it is possible to create an \SPD of depth $O(\log n)$ for every planar graph. \cite{AFGN22} showed that every graph which admits an \SPD of depth $k$, can be embedded into $\ell_1$ with distortion $O(\sqrt{k})$. In particular $O(\sqrt{\log n})$ for planar graphs.

In this paper we generalize the notion of \SPD by defining \emph{partial} \SPD (\PSPD). The difference is that in \PSPD we do not need all the vertices to be removed. That is, in \PSPD the last level of the hierarchy is allowed to be non-empty.
Given a planar graph $G$ with a terminal set $K$ and a face cover of size $\gamma(G,K)=\gamma$, using cycle separators \cite{Mil86} we create a \PSPD of depth $O(\log \gamma)$, such that for every cluster $C$ in the lower level of the hierarchy, all the remaining terminals $K\cap C$ lie on a single face. In other words, each such cluster is an Okamura-Seymour (O-S) graph.

We invoke the embedding of \cite{AFGN22} on our \PSPD, as a result we get an embedding with expansion $O(\sqrt{\log \gamma})$, where every pair of terminals $v,u$ that either was separated by the \PSPD, or lie close enough to some removed shortest path, has constant contraction.
The remaining work is to take care of terminal pairs that remained in the same cluster, and lie far from the cluster boundary.
As each such cluster is O-S graph, it embeds isometrically to $\ell_1$. However, we cannot simply embed each cluster independently of the entire graph. Such an oblivious embedding will create an unbounded expansion, as close-by pairs  might belong to different clusters.

Our solution, and the main technical part of the paper, is to create a truncated embedding \footnote{An embedding $f:V\rightarrow\ell_1$  is truncated if for every vertex $v$, $\|f(v)\|_1$ is bounded by some formerly specified number.}. Specifically, consider a cluster $C$ where all the terminals lie on a single face $F$. Let $\mathcal{B}=V\setminus C$ be the boundary of $C$, which is the set of vertices outside $C$.
We construct a Lipschitz embedding $f$ of $F$ into $\ell_1$ such that the norm $\|f(v)\|_1$ of every vertex $v\in F$ is bounded by its distance to the boundary $d_G(v,\mathcal{B})$, while $f$ has constant contraction for pairs far enough from the boundary. 
Our final embedding is defined as a concatenating of the embedding  for the \PSPD with a truncated embedding for every cluster, providing a constant contraction on all pairs and $O(\sqrt{\log \gamma})$ expansion.

Our truncated embedding does not use the embedding of \cite{OS81}. As a middle step, given a parameter $t>0$, we provide a uniformly truncated embedding \footnote{In uniformly truncated embedding $\|f(v)\|_1$ is bounded for all the vertices by a global single parameter.} $f_t$ such that $f_t$ is Lipschitz, the norm $\|f_t(v)\|_1$ of every vertex $v\in F$ is exactly $t$, and $f_t$ provides constant contraction for pairs at distance at most $t$. 
The construction of the  uniformly truncated embedding goes through a stochastic embedding into trees.
To create the non-uniformly truncated embedding we combine uniformly truncated embeddings for all possible truncation scales.

\subsection{Related Work}\label{subsec:related}
The notion of face cover $\gamma(G,K)$ was extensively studied in the context of Steiner tree problem \cite{EMV87,Ber90,KNvL20}, cuts and (multicommodity) flows \cite{MNS85,CW04}, all pairs shortest path \cite{Fre91,Fre95,CX00} and cut sparsifiers \cite{KR20,KPZ19}.
Given a drawing and a terminal set $K$, $\gamma(G,K)$ can be found in $2^{O(\gamma(G,K))}\cdot\poly(n)$ time, but generally it is known to be NP-hard \cite{BM88}. Frederickson \cite{Fre91} (Lemma 7.1) presented a polynomial-time approximation scheme (PTAS) for the problem of finding a face cover of minimum size. Specifically, given a planar graph with a drawing, Frederickson's algorithm finds a face cover of size at most $(1+\eps)\cdot\gamma(G,K)$ in $O(2^{\frac3\eps}\cdot n)$ time.
Denote by $\gamma^*(G,K)$ the minimal size of a face cover over all planar drawings of $G$. It is known that computing  $\gamma^*(G,K)$ is NP-hard \cite{BM88}. Frederickson \cite{Fre91} presented a $4$-approximation for $\gamma^*(G,K)$ in the special case where $K=V$, i.e. the terminals are the entire set $V$. However, for general $K\subseteq V$, to the best of the author's knowledge, no approximation is known. 

It is well known that Euclidean metrics, as well as distributions over trees, embed isometrically into $\ell_1$ (See \cite{matbook}). Therefore, in order to construct a bounded distortion embedding to $\ell_1$, it is enough to embed into either $\ell_2$ or a distribution over trees.
 
Outerplanar graphs are $1$-outerplanar. A graph is called $k$-outerplanar, if by removing all the vertices on the outer face, the graph becomes $k-1$-outerplanar.
Chekuri \etal \cite{CGNRS06} proved that $k$-outerplanar graphs embed into distribution over trees with $2^{O(k)}$ distortion.

Next consider minor-closed graph families. Following \cite{GNRS04}, Chakrabarti \etal~\cite{CJLV08} showed
that every graph with treewidth-$2$ (which excludes $K_4$ as a minor)
embeds into $\ell_1$ with distortion $2$ (which is tight, as shown by \cite{LR10}).
Already for treewidth-$3$ graphs, it is unknown whether they embed into $\ell_1$ with a constant distortion.
Abraham \etal \cite{AFGN22}  showed that every graph with pathwidth $k$ embeds into $\ell_1$ with distortion $O(\sqrt{k})$ (through $\ell_2$), improving a previous result of Lee and Sidiropoulos \cite{LS13} who showed a $(4k)^{k^3+1}$ distortion (via embedding into trees).
Graphs with treewidth $k$ are embeddable into $\ell_2$ with distortion $O(\sqrt{k\log n})$ \cite{KLMN04}.
For genus $g$ graphs, \cite{LS10} showed an embedding into Euclidean space with distortion $O(\log g+\sqrt{\log n})$.
Finally, for $H$-minor-free graphs, combining the results
of \cite{AGGNT19, KLMN04} provides Euclidean embeddings with
$O(\sqrt{|H|\log n})$ distortion.

The distortion guarantee in \Cref{cor:main} is somewhat similar to terminal distortion \cite{EFN17}, however there is also a guarantee on pairs of type $K\times X$. See \cite{EFN18,FGK24,EN22} for more on terminal distortion.
For other notions of distortion, Abraham \etal \cite{ABN11} showed that $\beta$-decomposable metrics (which include planar graphs as well as all other families mentioned in this section), for fixed $\beta$, embed into $\ell_2$ with scaling distortion $O(\sqrt{\log\frac1\eps})$. This means that for every $\eps\in(0,1)$ all but an $\eps$ fraction of the pairs in ${V\choose2}$ have  distortion at most $O(\sqrt{\log\frac1\eps})$.
Bartal \etal \cite{BFN19} proved that $\beta$-decomposable metrics (for fixed $\beta$) embed into $\ell_2$ with prioritized distortion $O(\sqrt{\log j})$. In more detail, given a priority order $\{v_1,\dots,v_n\}$ over the vertices, the pair $\{v_i,v_j\}$ for $j\le i$, will have distortion at most $O(\sqrt{\log j})$.

\subsection{Follow-Up Work}
In a followup work, Kumar \cite{Kumar22} showed that every planar graph can be embedded into $\ell_1$ with a non expansive embedding, such that every pair of vertices $u,v$ that lie on the same face have contraction at most $c$, for some universal constant $c$. Comparing to our work, the distortion in \cite{Kumar22} is constant regardless of the number of faces. However, the distortion guarantee is only for vertices lying  on the same face.

Since our paper, there been also exciting developments in embedding planar and minor-free graphs into bounded treewidth graphs.
Previously, it was known that every deterministic embedding of planar graphs into bounded-treewidth graph with constant distortion requires treewidth $\Omega(\sqrt{n})$ \cite{CG04}, while any stochastic embedding into graphs with constant treewidth  requires distortion distortion $\Omega(\log n)$ \cite{CG04,CJLV08} (a distortion which already can be obtained by embedding general metrics into trees \cite{FRT04}).
Furthermore, every stochastic embedding into graphs with treewidth $O(\frac{\log^{\nicefrac13}n}{\log\log n})$ requires expected distortion $\Omega(\frac{\log^{\nicefrac13}n}{\log\log n})$ \cite{FL22}. 
Recently, Cohen-Addad \etal \cite{CLPP23} showed that every fixed minor-free graph can be embedded with stochastic distortion $1+\eps$, into a distribution over graphs with treewidth $\poly(\frac1\eps,\log n, \log\Phi)$, where $\Phi=\frac{\max_{x,y}d_G(x,y)}{\min_{x,y}d_G(x,y)}$ is the aspect ratio.

Another new exciting line of work is regarding metric embeddings with additive stretch. Here we are given a graph $G$ with diameter $D$, and the guarantee is that the pairwise distances do not decrease, while they can increase by at most an $+\eps\cdot D$ additive factor.
Improving over \cite{FKS19}, Filtser and Le \cite{FL22}, and Chang \etal \cite{CCLMST23} showed that every planar graph can be deterministically embedded with additive  distortion $+\eps\cdot D$ into a graph with treewidth $O\left(\min\{\eps^{-1}\cdot(\log\log n)^2,\eps^{-4}\}\right)$.
Filtser and Le \cite{FL22} (improving over \cite{CFKL20}) showed that every $K_r$ minor-free graph can be stochastically embedded with expected additive distortion $+\eps\cdot D$ into a distribution of treewidth   $O_r(\eps^{-2}\cdot(\log\log n)^2)$ graphs.
There are also metric embeddings with alternative guarantees of additive distortion (clan, Ramsey) \cite{FL21,FL22}.

\section{Preliminaries}
\emph{Graphs.} We consider connected undirected graphs $G=(V,E)$ with edge weights
$w: E \to \R_{\ge 0}$. Let $d_{G}$ denote the shortest path metric in
$G$. For a vertex $x\in
V$ and a set $A\subseteq V$, let $d_{G}(x,A):=\min_{a\in A}d(x,a)$,
where $d_{G}(x,\emptyset):= \infty$.  For a subset of vertices
$A\subseteq V$, let $G[A]$ denote the induced graph on $A$. 
Let $G\setminus A := G[V\setminus A]$ be the graph after deleting the vertex set $A$ from $G$.

	See \Cref{sec:intro} for definitions of \emph{embedding, distortion, contraction}, \emph{expansion} and \emph{Lipschitz}. We say that an embedding is \emph{dominating} if it is non-contractive. Given a graph family $\mathcal{F}$, a \emph{stochastic embedding} of $G$ into $\mathcal{F}$ is a distribution $\mathcal{D}$ over pairs $(H,f_H)$ where $H\in\mathcal{F}$ and $f_H$ is an embedding of $G$ into $H$. We say that $\mathcal{D}$ is dominating if for every $(H,f_H)\in\supp(\mathcal{D})$, $f_H$ is dominating.
	We say that a dominating stochastic embedding $\mathcal{D}$ has expected distortion $t$, if for every pair $u,v\in V$ it holds that 
	\[
	\mathbb{E}_{(H,f_{H})\sim\mathcal{D}}\left[d_{H}\left(f_H(u),f_H(v)\right)\right]\le t\cdot d_{G}(u,v)~.
	\]

\begin{wrapfigure}{r}{0.23\textwidth}
	\begin{center}
		\vspace{-20pt}
		\includegraphics[width=0.22\textwidth,page=2]{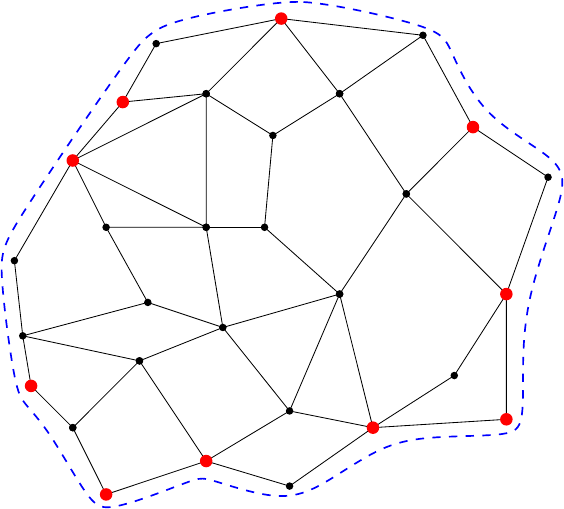}
		\vspace{-7pt}
	\end{center}
	\vspace{-10pt}
\end{wrapfigure}
A \emph{terminated} planar graph $G=(V,E,w,K)$ is a planar graph $(V,E,w)$, with a subset of terminals $K\subseteq V$.
A graph G is outerplanar if there is a drawing of $G$ in the plane such that all the vertices lie on the unbounded face.
A \emph{face cover} is a set of faces such that every terminal lies on at least one face from the cover.
Given a graph $G$ with a drawing in the plane, denote by $\gamma(G,K)$ the minimal size of a face cover.
In the special case where all the terminals are covered by a single face, i.e. $\gamma(G,K)=1$, we say that $G$ is an Okamura-Seymour graph, or O-S graph for short. See illustration on the right (using the notation from \Cref{fig:facecover}).

A \emph{tree decomposition} of a graph $G=(V,E)$ is a tree
$T$ with nodes $B_1,\dots,B_s$ (called \emph{bags}) where each $B_i$ is a subset of $V$ such that:
(1) For every edge $\{u,v\}\in E$, there is a bag $B_i$ containing both $u$ and $v$. 
(2) For every vertex $v\in V$, the set of bags containing $v$ forms a connected subtree of $T$. (3) Every vertex belongs to at least one bag.
The \emph{width} of a tree decomposition is $\max_i\{|B_i|-1\}$. The \emph{treewidth} of $G$ is the minimal
width of a tree decomposition of $G$.
It is straightforward to verify that every tree graph has treewidth $1$.

Given a set of $s$ embeddings $f_i:V\rightarrow \R^{d_i}$ for $i\in\{1,\dots,s\}$, the \emph{concatenation} of $f_1,\dots,f_s$, denoted by $\oplus_{i=1}^s f_i$, is a function $f:V\rightarrow\R^{\sum_id_i}$, where the coordinates from $1$ to $d_1$ correspond to $f_1$, the coordinates from $d_1+1$ to $d_1+d_2$ correspond to $f_2$, etc.

\section{Partial Shortest Path Decomposition}
Abraham \etal \cite{AFGN22} defined \emph{shortest path decompositions} (\SPDs) of ``low depth''.
Every (weighted) path graph has an \SPDdepth $1$. A graph $G$ has an \SPDdepth
$k$ if there exists a \emph{shortest path} $P$, such that every connected
component in $G\setminus P$ has an \SPDdepth $k-1$.
In other words, given a graph, in \SPD we hierarchically delete shortest paths from each connected component, until no vertices remain.
In this paper we define a generalization called \emph{partial shortest path decomposition} (\PSPD), where we remove the requirement that all the vertices will be deleted. See the formal definition below. 
In \Cref{sec:PSPDconstruction} we will argue that every terminated planar graph with face cover of size $\gamma$ has a \PSPD of depth $O(\log \gamma)$ such that in each connected component in the lower level of the hierarchy, all the terminals lie on a single face.

A \emph{partial  partition} $\mathcal{X}$ of a set $X$, is a disjoint set of subsets of $X$. In other words, for every $A\in \mathcal{X}$, $A\subseteq X$, and for every different subsets $A,B\in\mathcal{X}$, $A\cap B=\emptyset$.
\begin{definition}[Partial Shortest Path Decomposition (\PSPD)]\label{def:P_SPD}
	Given a weighted graph $G=(V,E,w)$, a \PSPD of depth $k$ is a pair
	$\left\{ \mathcal{X},\mathcal{P}\right\} $, where $\mathcal{X}$ is a
	collection $\mathcal{X}_{1},\dots,\mathcal{X}_{k+1}$ of partial partitions of
	$V$, 
	$\mathcal{P}$ is a collection of sets of paths
	$\mathcal{P}_{1},\dots,\mathcal{P}_{k}$, and:
	\begin{enumerate}
		\item  $\mathcal{X}_1=\{V\}$.
		\item For every $1\leq i\leq k$ and every cluster $X\in\mathcal{X}_{i}$,
		there exist a unique path $P_X\in\mathcal{P}_{i}$ such that $P_{X}$
		is a shortest path in $G[X]$.
		\item For every $2\leq i\leq k+1$, $\mathcal{X}_i$ consists of all
		connected components of $G[X\setminus P_{X}]$ over all $X\in\mathcal{X}_{i-1}$.
	\end{enumerate}
	The remainder of the \PSPD $\left\{ \mathcal{X},\mathcal{P}\right\} $ is a pair $\left\{ \mathcal{C},\mathcal{B}\right\}$ where $\mathcal{C}=\mathcal{X}_{k+1}$ is the set of connected components in the final level of the \PSPD, and $\mathcal{B}=\bigcup_{i=1}^{k}\cup\mathcal{P}_{i}$ is the set of all the vertices in the removed paths. $\mathcal{B}$ is also called the boundary.
\end{definition}

	\begin{figure}[t]
	\centering{\includegraphics[scale=0.61,page=8]{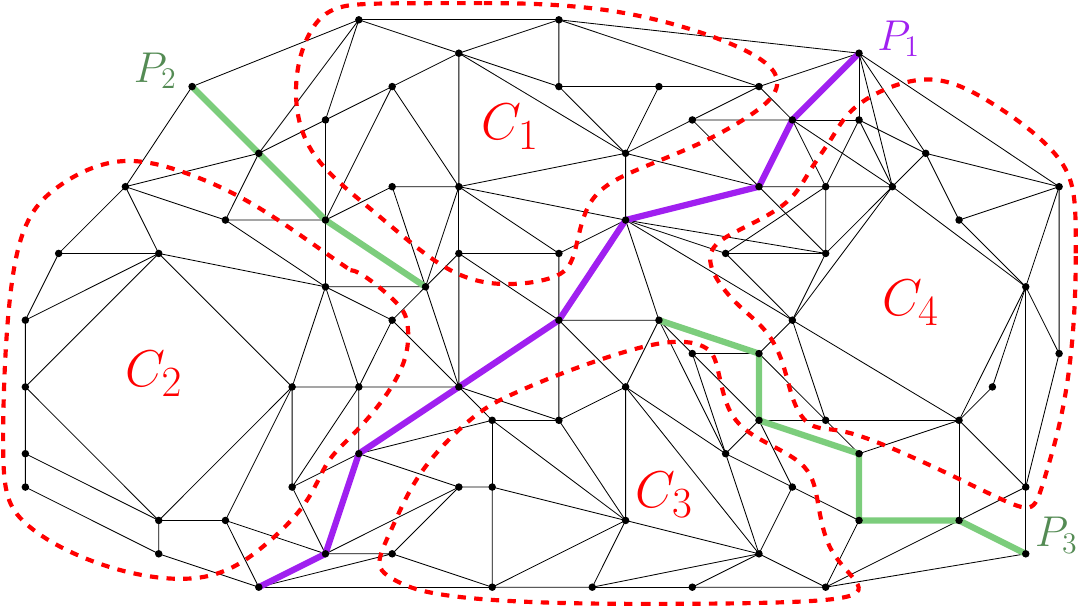}} 
	\caption{\label{fig:PSPD}\small \it 
		Illustration of a \PSPD of depth $2$.
		Here $\mathcal{P}_{1}$ is the purple path $P_1$, while $\mathcal{P}_{2}$ is the two green paths $P_2$ and $P_3$. 	The reminder of the \PSPD is the $4$ clusters $\mathcal{C}=\{C_1,C_2,C_3,C_4\}$ encircled by a dashed red line, while the boundary is $\mathcal{B}=P_1\cup P_2\cup P_3$.
	}
\end{figure}

Under \Cref{def:P_SPD}, \SPD is a special case of \PSPD where $\mathcal{C}=\emptyset$ (and $\mathcal{B}=V$). The main theorem in \cite{AFGN22} states that if a graph $G$ has \SPD of depth $k$, then it is embeddable into $\ell_1$ with distortion $O(\sqrt{k})$ \footnote{In fact \cite{AFGN22} proved a more general result, stating that $G$ is embeddable into $\ell_p$ for $p\in[1,\infty]$, with distortion $O(k^{\min\{\frac12,\frac1p\}})$. Similarly, in \Cref{thm:EmbPSPD} can replace $\ell_1$ with $\ell_p$ and the expansion $\sqrt{k}$ with $k^{\min\{\frac12,\frac1p\}}$. The contraction condition and values remains the same.}.
\cite{AFGN22} construct a different embedding for each level of the decomposition. Each such embedding is Lipschitz, while for every pair of vertices $u,v\in V$ there is some level $i$ such that the embedding for this level has constant contraction w.r.t. $u,v$. 
Specifically, the level  $i$  with the bounded contraction guarantee is the first level in which either $u,v$ are separated or the distance between $\{u,v\}$ to a deleted path is at most $\frac{d_G(u,v)}{12}$. 
In particular, given a \PSPD, by using the exact same embedding from \cite{AFGN22} (w.r.t. the existing levels in the decomposition), we get the following theorem.

\begin{theorem}[Embedding using \PSPD]\label{thm:EmbPSPD}
	Let $G=(V,E,w)$ be a weighted graph, and let $\left\{ \mathcal{X},\mathcal{P}\right\} $ be a \PSPD of depth $k$ with remainder  $\left\{ \mathcal{C},\mathcal{B}\right\}$. There is an embedding $f:V\rightarrow\ell_1$ with the following properties:
	\begin{enumerate}
		\item For every $u,v\in V$, $\left\Vert f(v)-f(u)\right\Vert _{1}\le O(\sqrt{k})\cdot d_G(u,v)$.
		\item For every $u,v\in V$ which are either separated by $\mathcal{C}$ (that is $u,v$ do not belong to the same cluster in $\mathcal{C}$), or such that $\min\left\{ d_{G}(v,\mathcal{B}),d_{G}(u,\mathcal{B})\right\} \le \frac{d_G(u,v)}{12}$, it holds that		
		$\left\Vert f(v)-f(u)\right\Vert _{1}\ge d_G(u,v)$.
	\end{enumerate}
	
\end{theorem}

\section{Uniformly Truncated Embedding}\label{sec:UniformTruncated}
In this section we construct a uniformly truncated embedding for O-S graphs into $\ell_1$. Specifically, given a truncation parameter $t$, we show how to embed O-S graphs into $\ell_1$ via a Lipschitz map such that the norm of all the vectors is exactly $t$, and it is non-contractive for terminals at distance at most $t$.
We will use two  previous results on stochastic embeddings. The following theorem was proven by Englert \etal \cite{EGKRTT14} (Thm. 12) in a broader sense. Lee \etal \cite{LMM15} (Thm. 4.4) observed that it implies embedding of O-S graphs into outerplanar graphs.
\begin{theorem}\label{thm:OS_embedability}
	Consider a weighted planar graph $G=(V,E,w)$ with $F\subseteq V$ being a face. There is a stochastic embedding of $F$ into dominating outerplanar graphs with expected distortion $O(1)$.
\end{theorem}
The following theorem was proven by Gupta \etal \cite{GNRS04} (Thm. 5.4).
\begin{theorem}\label{thm:OuterPlanarToTrees}
	Consider a weighted outerplanar graph $G=(V,E,w)$. There is a stochastic embedding of $G$ into dominating trees with expected distortion $O(1)$.
\end{theorem}
As it was already observed in \cite{KLR19}, we conclude:
\begin{restatable}[]{corollary}{OStoTrees}
	\label{cor:OStoTrees}
	Consider a planar graph $G=(V,E,w)$ with a face $F$.
	There is a stochastic embedding of $F$ into dominating trees with expected distortion $O(1)$.
\end{restatable}
\begin{proof}
	Let $\mathcal{D}_{OP}$ be the distribution over outerplanar graphs from \Cref{thm:OS_embedability}.
	For every $G'\in\supp(\mathcal{D}_{OP})$, let $\mathcal{D}_{G'}$ be the distribution over trees from \Cref{thm:OuterPlanarToTrees} (w.r.t. $G'$).
	We define a distribution $\mathcal{D}$ of embeddings of $F$ into trees as follows. First sample an outerplanar graph $G'$ using $\mathcal{D}_{OP}$. Then sample a tree $T$ using $\mathcal{D}_{G'}$. We argue that the distribution $\mathcal{D}$ has the desired properties.
	
	All the embeddings in the support of $\mathcal{D}$ are dominating metrics as both  $\mathcal{D}_{OP}$ and $\mathcal{D}_{G'}$ (for all $G'$) have this property. Consider $v,u\in V$, it holds that 	
	\begin{align*}
		\mathbb{E}_{T'\sim\mathcal{D}}\left[d_{T}(v,u)\right] & =\sum_{G'\in\supp(\mathcal{D}_{OP})}\Pr_{\mathcal{D}_{OP}}\left[G'\right]\cdot\left(\sum_{T\in\supp(\mathcal{D}_{G'})}\Pr_{\mathcal{D}_{G'}}\left[T\mid G'\right]\cdot d_{T}(v,u)\right)\\
		& =\sum_{G'\in\supp(\mathcal{D}_{OP})}\Pr_{\mathcal{D}_{OP}}\left[G'\right]\cdot O\left(d_{G'}(v,u)\right)~~=~~O\left(d_{G}(v,u)\right)~.
	\end{align*}
\end{proof}

A first step towards truncated embedding of O-S graphs will be a truncated embedding of trees.
\begin{lemma}\label{lem:TreeToBoundL1}
	Let $T=(V,E,w)$ be some tree and let $t>0$ be a truncation parameter. There exists an embedding $f:T\rightarrow\ell_1$ such that the following holds:
	\begin{enumerate}
		\item Sphere surface: for every $v\in V$, $\left\Vert f(v)\right\Vert _{1}=t$.
		\item Lipschitz: for every $u,v\in V$, $\left\Vert f(v)-f(u)\right\Vert _{1}\le 4\cdot d_{T}(v,u)$.
		\item Bounded contraction: for every $u,v\in V$, $\left\Vert f(v)-f(u)\right\Vert _{1}\ge\min\left\{ d_{T}(v,u),t\right\}$.
	\end{enumerate}
\end{lemma}
\begin{proof}
	Add a new vertex $v_t$ to $T$ with edges of weight $\frac t2$ to all the other vertices. Call the new graph $T_t$. 
	Notice that $T_t$ has treewidth $2$.
	According to Chakrabarti \etal~\cite{CJLV08}, there is an embedding $f_{\tw}$ of $T_t$ into $\ell_1$ with  distortion $2$. By rescaling, we can assume that the contraction is $1$ (and the expansion is at most $2$). Additionally, by shifting, we can assume that $f_{\tw}(v_t)=\vec{0}$. Let $f$ be the embedding $f_{\tw}$ with an additional coordinate. The value of every $v\in V$ in the new coordinate equals to $t-\|f_{\tw}(v)\|_1$.
	We argue that $f$ has the desired properties.
	
	The first property follows as for every vertex $v\in V$, the distance in $T$ to $v_t$ is exactly $\frac t2$, therefore $\left\Vert f_{\tw}(v)\right\Vert _{1}=\left\Vert f_{\tw}(v)-f_{\tw}(v_{t})\right\Vert _{1}\le 2\cdot d_{T_t}(v,v_{t})=t$. Therefor $\left\Vert f(v)\right\Vert _{1}=\left\Vert f_{\tw}(v)\right\Vert _{1}+\left|t-\left\Vert f_{\tw}(v)\right\Vert _{1}\right|=t$.
	
	The second property follows as $f_{\tw}$ has expansion $2$, and distances in $T_t$ can only decrease w.r.t. distances in $T$. Thus for every  $u,v\in V$, $\left\Vert f_{\tw}(v)-f_{\tw}(u)\right\Vert _{1}\le 2\cdot d_{T_{t}}(v,u)\le 2\cdot d_{T}(v,u)$.
	By the triangle inequality,
	\begin{align*}
	\left\Vert f(v)-f(u)\right\Vert _{1} & =\left\Vert f_{\tw}(v)-f_{\tw}(u)\right\Vert _{1}+\left|\left(t-\left\Vert f_{\tw}(v)\right\Vert _{1}\right)-\left(t-\left\Vert f_{\tw}(u)\right\Vert _{1}\right)\right|\\
	& \le2\cdot\left\Vert f_{\tw}(v)-f_{\tw}(u)\right\Vert \le4\cdot d_{T}(v,u)~.
	\end{align*}
	
	For the third property, consider some pair $u,v\in V$. 
	As every shortest path containing the new vertex $v_t$ will be of weight at least $t$, it holds that $d_{T_{t}}(v,u)=\min\{d_{T}(v,u),t\}$. We conclude that $\left\Vert f(v)-f(u)\right\Vert _{1}\ge\left\Vert f_{\tw}(v)-f_{\tw}(u)\right\Vert _{1}\ge d_{T_{t}}(v,u)=\min\{d_{T}(v,u),t\}$.
\end{proof}

Next, we construct an embedding of O-S graphs into the sphere of radius $t$ in $\ell_1$.
\begin{corollary}\label{cor:OStoBoundL1}
	
	Let $G=(V,E,w)$ be a planar graph, $F$ a face, and $t>0$ a truncation parameter. There exists an embedding $f:F\rightarrow\ell_1$ such that the following holds:
	\begin{enumerate}
		\item Sphere surface: for every $v\in F$, $\left\Vert f(v)\right\Vert _{1}= t$.
		\item Lipschitz: for every $u,v\in F$, $\left\Vert f(v)-f(u)\right\Vert _{1}\le O(d(v,u))$.
		\item Bounded Contraction: for every $u,v\in F$, $\left\Vert f(v)-f(u)\right\Vert _{1}\ge\min\left\{d_{G}(v,u),t\right\}$.
	\end{enumerate}
\end{corollary}
\begin{proof}
	Let $\mathcal{D}$ be the distribution over dominating trees guaranteed in \Cref{cor:OStoTrees}.
	For every $T\in\supp(\mathcal{D})$, let $f_T$ be the embedding of $T$ into $\ell_1$ from \Cref{lem:TreeToBoundL1} with parameter $t$.
	Our embedding is constructed by concatenating all $f_T$, scaled by their probabilities. That is, $f=\oplus\left\{ \Pr\left[T\right]\cdot f_{T}\mid T\in\supp(\mathcal{D})\right\}$.
	
	The first property follows as for every $v\in X$ and $f_T$,  $\left\Vert f(v)\right\Vert _{1}= t$. 
	Similarly, the third property follows as for every $u,v\in F$ and $T\in\supp(\mathcal{D})$, $\left\Vert f_{T}(v)-f_{T}(u)\right\Vert _{1}\ge\min\left\{ d_{T}(v,u),t\right\} \ge\min\left\{ d(v,u),t\right\}$.	
	
	The second property follows as for every $v,u\in V$,
	\begin{align*}
	\left\Vert f(v)-f(u)\right\Vert _{1} & =\sum_{T}\Pr\left[T\right]\cdot\left\Vert f_{T}(v)-f_{T}(u)\right\Vert _{1}\\
	& \le\sum_{T}\Pr\left[T\right]\cdot 4\cdot d_{T}(v,u) ~=~ 4\cdot \E_{T\sim\mathcal{D}}\left[d_{T}(v,u)\right]~=~O( d_G(v,u))~.
	\end{align*}
\end{proof}
\begin{remark}\label{remark:time}
	Efficient construction: the support of the distribution $\mathcal{D}$ might be of exponential size. Nevertheless, we can bypass this barrier by carefully sampling polynomially many trees.\\
	Denote by $m$ the number of vertices on $F$. For a pair $v,u$, by Markov inequality, the probability that a sampled tree has distortion larger than $m^3$ on $u,v$ is $O(m^{-3})$. We say that a tree is bad if it has distortion $m^3$ on some pair, otherwise it is good. By the union bound, the probability for sampling a bad tree is $O(1/m)$.
	Let $\mathcal{D}'$ be the distribution $\mathcal{D}$ restricted to good trees only. $\mathcal{D}'$ is a distribution over dominating trees with constant expected distortion, and worst case distortion $m^3$. Sample $m^6$ trees $T_1,\dots,T_{m^6}$ from $\mathcal{D}'$. By Hoeffding \footnote{See \url{https://sarielhp.org/misc/blog/15/09/03/chernoff.pdf}, Theorem 7.4.3 .}
	and union bound inequalities, w.h.p. the average distortion of all pairs will be constant. Define the embedding $f=\oplus\left\{ m^{-6}\cdot f_{T_i}\right\}_{i=1}^{m^6}$. The proof above still goes through.	 
\end{remark}

\section{Non-Uniformly Truncated Embedding}\label{sec:NonUniformTruncated}
In this section we generalize \Cref{cor:OStoBoundL1}. Instead of a uniform truncation parameter $t$ for all the vertices, we will allow a somewhat customized truncation. See \Cref{fig:truncatedOS} for illustration of the guarantees in the lemma. 
\begin{lemma}\label{lem:BoundedOStoL1}
	Let $G=(V,E,w)$ be a planar graph with a given drawing on the plane. Let $F,\mathcal{I},\mathcal{B}\subseteq V$ such that $F\subseteq \mathcal{I}$, $\mathcal{I}\cup\mathcal{B}=V$, $\mathcal{I}\cap\mathcal{B}=\emptyset$, and $F$ is a face in $G[\mathcal{I}]$.
	Then there is an embedding $f:F\rightarrow\ell_1$ such that the following holds:
	\begin{enumerate}
		\item For every $v\in F$, $\|f(v)\|_1=d_G(v,\mathcal{B})$.
		\item Lipschitz: for every $u,v\in F$, $\left\Vert f(v)-f(u)\right\Vert _{1}\le O(d_G(v,u))$.
		\item Bounded Contraction: for every $c\ge 1$, and every $u,v\in F$ such that  $\min\left\{ d_{G}(v,\mathcal{B}),d_{G}(u,\mathcal{B})\right\} \ge \frac{d_G(u,v)}{c}$ it holds that $\left\Vert {f}(v)-{f}(u)\right\Vert _{1}\ge\frac{d_{G}(v,u)}{12\cdot c}$.
	\end{enumerate}
\end{lemma}

	\begin{figure}[t]
	\centering{\includegraphics[scale=0.55]{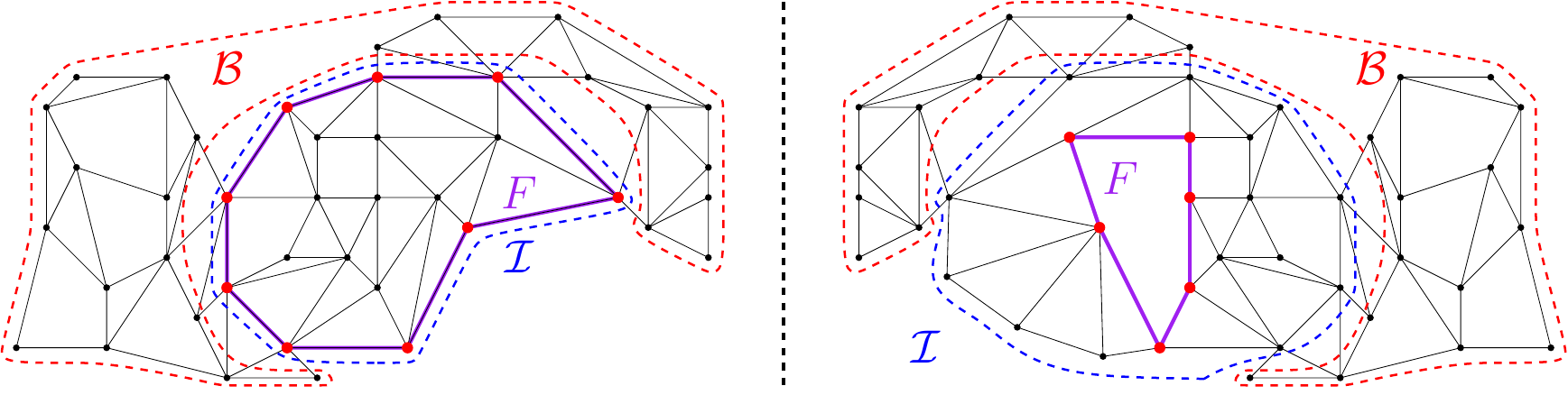}} 
	\caption{\label{fig:truncatedOS}\small \it 
		Both figures on the left and on the right illustrate two different instances where \Cref{lem:BoundedOStoL1} can be applied. The boundary is $\mathcal{B}$ encircled by a red dashed line. The interior $\mathcal{I}$, which is the rest of $G$ vertices, is encircled by a blue dashed line.
		The induced graph on the interior $G[\mathcal{I}]$ contains a face $F$ (denoted by a purple line, while $F$ vertices are red). \Cref{lem:BoundedOStoL1} constructs a Lipschitz (guarantee (2)) embedding of $F$ vertices into $\ell_1$, where the norm of the image of each vertex $v\in F$ equals to its distance to the boundary (guarantee (1)). Furthermore, for every $u,v\in F$ it holds that $\left\Vert {f}(v)-{f}(u)\right\Vert _{1}\ge\frac{\min\left\{ d_{G}(v,\mathcal{B}),d_{G}(u,\mathcal{B})\right\} }{12}$ (guarantee (3)).
	}
\end{figure}

\begin{proof}
We will construct the non-uniformly truncated embedding by a smooth combination of uniformly truncated embeddings for all possible truncation scales. A similar approach was applied in \cite{AFGN22}.=
Assume (by scaling) w.l.o.g. that the minimal weight of an edge in $G$ is $1$. Let $M\in \mathbb{N}$ be minimal integer such that the diameter
of $G$ is strictly bounded by $2^M$. 

Consider the graph $G[\mathcal{I}]$ induced by $\mathcal{I}$. Note that $G[\mathcal{I}]$  is an O-S graph w.r.t. $F$.
For every distance scale $t\in\{0,1,\dots,M\}$ let $f_t$ be the embedding of $F$ w.r.t. the shortest path metric induced by $G[\mathcal{I}]$ from \Cref{cor:OStoBoundL1} with truncation parameter $2^t$.
For a vertex $v\in \mathcal{I}$, let $t_{v}\in\mathbb{N}$ be such that $d_{G}(v,\mathcal{B})\in[2^{t_{v}},2^{t_{v}+1})$. Set $\lambda_v=\frac{	d_{G}(v,\mathcal{B})-2^{t_v}}{2^{t_v}}$.  Note that $0\le
\lambda_v<1$.
For $t\in\{0,\dots,M\}$ we define a function ${\tilde{f}}_{t}:F\rightarrow\ell_1$,
\begin{gather*}
{\tilde{f}}_{t}(v)\doteq\begin{cases}
\lambda_{v}\cdot{f}_{t}(v)\qquad\qquad & \text{if }t=t_{v}+1,\\
(1-\lambda_{v})\cdot{f}_{t}(v) & \text{if }t=t_{v},\\
\vec{0} & \text{otherwise}.
\end{cases}
\label{eq:ftilde}
\end{gather*}
Define $f$ to be the concatenation of ${\tilde{f}}_{0}(v),\dots,{\tilde{f}}_{M}(v)$.

For every $v\in F$, according to \Cref{cor:OStoBoundL1} it holds that 
\begin{align}
\left\Vert {f}(v)\right\Vert _{1} & =(1-\lambda_{v})\cdot\left\Vert {\tilde{f}}_{t_{v}}(v)\right\Vert _{1}+\lambda_{v}\cdot\left\Vert {\tilde{f}}_{t_{v+1}}(v)\right\Vert _{1}\nonumber\\
& =(1-\lambda_{v})\cdot2^{t_{v}}+\lambda_{v}\cdot2^{t_{v}+1}\nonumber\\
& =\frac{2^{t_{v}+1}-d_{G}(v,\mathcal{B})}{2^{t_{v}}}\cdot2^{t_{v}}+\frac{d_{G}(v,\mathcal{B})-2^{t_{v}}}{2^{t_{v}}}\cdot2^{t_{v}+1}~~=~~d_{G}(v,\mathcal{B})~.
\label{eq:boundedNorm}
\end{align}
Next we prove that $f$ is Lipschitz. Consider a pair of vertices $u,v\in \mathcal{F}$. 
If $d_G(u,v)<d_{G[\mathcal{I}]}(u,v)$, then the shortest path from $u$ to $v$ in $G$ has to go through the boundary $\mathcal{B}$. It follows that $d_G(u,\mathcal{B})+d_G(v,\mathcal{B}) \le d_G(u,v)$. We conclude
\begin{align}
\left\Vert {f}(v)-{f}(u)\right\Vert _{1}  \le\left\Vert {f}(v)\right\Vert _{1}+\left\Vert {f}(u)\right\Vert _{1}\nonumber \overset{(\ref{eq:boundedNorm})}{=} d_{G}(v,\mathcal{B})+d_{G}(u,\mathcal{B})~\le~d_{G}(v,u)~.\label{eq:fCexpanionBound}
\end{align}

Otherwise, $d_G(u,v)=d_{G[\mathcal{I}]}(u,v)$. It follows from \Cref{cor:OStoBoundL1} that for every scale parameter $t$,
$\left\Vert {f}_{t}(v)-{f}_{t}(u)\right\Vert \le O\left(d_{G[ \mathcal{I}]}(u,v)\right)=O\left(d_{G}(u,v)\right)$. 
We will prove a similar inequality for $\tilde{f}_{t}$. As $\tilde{f}_{t}(u)$ and $\tilde{f}_{t}(v)$ combined might be nonzero in at most $4$ different scales, the bound on expansion will follow.

Denote by $p_t$ the scaling factor of $v$ in ${\tilde{f}}_{t}(v)$. That is, $p_{t_v+1}=\lambda_v$, $p_{t_v}=1-\lambda_v$, and $p_t=0$ for $t\notin\{t_v,t_{v}+1\}$. Similarly, define $q_t$ for $u$.
First, observe that for every $t$,
\begin{align*}
\left\Vert \tilde{f}_{t}(v)-\tilde{f}_{t}(u)\right\Vert _{1} & =\left\Vert p_{t}\cdot f_{t}(v)-q_{t}\cdot f_{t}(u)\right\Vert _{1}\\
& \le\min\left\{ p_{t},q_{t}\right\} \cdot\left\Vert f_{t}(v)-f_{t}(u)\right\Vert _{1}+\left|p_{t}-q_{t}\right|\cdot\max\left\{ \left\Vert f_{t}(v)\right\Vert _{1},\left\Vert f_{t}(u)\right\Vert _{1}\right\} \\
& \le O(d_{G}(u,v))+\left|p_{t}-q_{t}\right|\cdot2^{t}~.
\end{align*}

It suffice to show that
$\left|p_{t}-q_{t}\right|=O(d_G(u,v)/2^{t})$. Indeed, for indices
$t\notin\left\{ t_{u},t_{u}+1,t_{v},t_{v}+1\right\} $, $p_{t}=q_{t}=0$,
and in particular $\left|p_{t}-q_{t}\right|=0$. Let us consider the other cases.
W.l.o.g. assume that $d_{G}\left(v,\mathcal{B}\right)\ge
d_{G}\left(u,\mathcal{B}\right)$ and hence $t_{v}\ge t_{u}$. We proceed by case analysis.
\begin{itemize}
	\item $\boldsymbol{t_{u}=t_{v}:}$ In this case,
	$\left|p_{t_{v}}-q_{t_{v}}\right|=\left|(1-\lambda_{v})-(1-\lambda_{u})\right|
	=|\lambda_{v}-\lambda_{u}|=
	\left|p_{t_{v}+1}-q_{t_{v}+1}\right|$. The value of this quantity is bounded by
	\begin{align*}
	\lambda_{v}-\lambda_{u}  ~=~\frac{d_{G}\left(v,\mathcal{B}\right)-2^{t_{v}}}{2^{t_{v}}}-\frac{d_{G}\left(u,\mathcal{B}\right)-2^{t_{v}}}{2^{t_{v}}}~=~\frac{d_{G}\left(v,\mathcal{B}\right)-d_{G}\left(u,\mathcal{B}\right)}{2^{t_v}}~\le~\frac{d_{G}(u,v)}{2^{t_{v}}}~.
	\end{align*}
	Hence, we get that $\left|p_{t}-q_{t}\right|=O(d_G(u,v)/2^t)$ for
	all $t \in \{t_v, t_v+1\}$.
	
	\item $\boldsymbol{t_{u}=t_{v}-1}:$ It holds that
	
	\begin{align*}
		\lambda_{v}+(1-\lambda_{u}) & ~=~\frac{d_{G}\left(v,\mathcal{B}\right)-2^{t_{v}}}{2^{t_{v}}}+\frac{2^{t_{u}+1}-d_{G}\left(u,\mathcal{B}\right)}{2^{t_{u}}}\\
		& ~\le~\frac{\left(d_{G}\left(v,\mathcal{B}\right)-2^{t_{v}}\right)+\left(2^{t_{u}+1}-d_{G}\left(u,\mathcal{B}\right)\right)}{2^{t_{u}}}~\le~\frac{d_{G}(u,v)}{2^{t_{u}}}~.		
	\end{align*}
	We conclude:
	\[
	\begin{array}{ccccc}
		\left|p_{t_{v}+1}-q_{t_{v}+1}\right| & = & \lambda_{v} & = & O(d_{G}(u,v)/2^{t_{v}})\\
		\left|p_{t_{v}}-q_{t_{v}}\right| & = & \left|1-\lambda_{v}-\lambda_{u}\right| & = & O(d_{G}(u,v)/2^{t_{v}})\\
		\left|p_{t_{u}}-q_{t_{u}}\right| & = & 1-\lambda_{u} & = & O(d_{G}(u,v)/2^{t_{v}})
	\end{array}~~.
	\]
	
	\item $\boldsymbol{t_{u}<t_{v}-1}:$ By the definition of $t_v$ and
	$t_u$,
	$d_{G}(v,u)\ge d_{G}(v,\mathcal{B})-d_{G}(u,\mathcal{B})\ge2^{t_{v}}-2^{t_{u}+1}\ge2^{t_{v}-1}$. It follows that for every $t\le t_v+1$,
	$\left|p_{t}-q_{t}\right|\le1\le\frac{d_G(u,v)}{2^{t_{v}-1}} =
	O\left(\frac{d_G(u,v)}{2^{t}}\right)$.
\end{itemize}	

Next we argue that $f$ has small contraction for pairs far enough form the boundary. Consider a pair of vertices $u,v\in F$, and let $c\ge 1$ such that $\min\left\{ d_{G}(v,\mathcal{B}),d_{G}(u,\mathcal{B})\right\} \ge \frac{d_G(u,v)}{c}$.
It holds that $2^{t_v}> \frac12\cdot d_{G}(v,\mathcal{B})\ge \frac{1}{2c}\cdot d_G(u,v)$.
For every $t\ge t_v$, by the contraction property of \Cref{cor:OStoBoundL1}, it holds that  
\begin{equation}
\left\Vert f_{t}(v)-f_{t}(u)\right\Vert_1 \ge\min\left\{ d_{G[\mathcal{I}]}(v,u),2^{t}\right\} \ge\frac{d_{G}(v,u)}{2c}~,\label{eq:NoScaleBOund}
\end{equation}
where the last inequality follows as $c\ge 1$.
Assume w.l.o.g. that  $d_{G}(v,\mathcal{B})\ge d_{G}(u,\mathcal{B})$, thus
$t_v\ge t_u$ and let $t\in\{t_v,t_{v}+1\}$ such that $p_t\ge \frac12$.
Set $S=12\cdot c$. We consider two cases: 
\begin{itemize}
	\item If $|p_t-q_t|>\frac{d_G(v,u)}{2^t\cdot S}$ then
	\begin{align*}
	\left\Vert {f}(v)-{f}(u)\right\Vert _{1} & \ge\left\Vert p_{t}\cdot f_{t}(v)-q_{t}\cdot f_{t}(u)\right\Vert _{1}\\
	& \ge\left|\left\Vert p_{t}\cdot f_{t}(v)\right\Vert _{1}-\left\Vert q_{t}\cdot f_{t}(u)\right\Vert _{1}\right|\\
	& =\left|p_{t}-q_{t}\right|\cdot2^{t}>\frac{d_G(v,u)}{S}~.
	\end{align*}
	\item Else $|p_t-q_t|\le\frac{d_G(v,u)}{2^t\cdot S}$. First, assume that  $p_t\ge q_t$. As $\frac{d_G(v,u)}{2^t\cdot S}\le\frac{2c}{S}= \frac16$, it holds that $q_t\ge \frac13$. We conclude \begin{align*}
	\left\Vert {f}(v)-{f}(u)\right\Vert _{1} & \ge\left\Vert p_{t}\cdot f_{t}(v)-q_{t}\cdot f_{t}(u)\right\Vert _{1}\\
	& \ge q_{t}\cdot\left\Vert  f_{t}(v)- f_{t}(u)\right\Vert _{1}-\left|p_{t}-q_{t}\right|\cdot\left\Vert  f_{t}(v)\right\Vert \\
	& \overset{(\ref{eq:NoScaleBOund})}{\ge}\frac{1}{3}\cdot\frac{d_{G}(v,u)}{2c} -\frac{d_G(v,u)}{S}~=~\frac{d_G(v,u)}{S}~.
	\end{align*}	
	The case where $q_t>p_t$ is symmetric.
\end{itemize}
\end{proof}

\begin{remark}
	In \Cref{lem:BoundedOStoL1} we used uniformly truncated embeddings in order to create a non-uniformly truncated embedding. Such a transformation might be relevant in other contexts as well. 
	Consider a case where each vertex $v$ has some truncation parameter $s_v$, and there exists a uniformly truncated embedding for every parameter $t$.
	As long as for every $v,u$, $|s_v-s_u|= O(d_G(u,v))$, following the same construction as above, one can create a similar non-uniformly truncated embedding where $\|f(v)\|=s_v$.
\end{remark}

\section{Construction of a \PSPD for Planar Graphs}\label{sec:PSPDconstruction}
We begin this section by proving the following separator theorem.  Even though similar statements to \Cref{thm:facesSeperator} already appeared in the literature, we provide a proof for completeness.
\begin{restatable}[]{theorem}{faceSeprator}\label{thm:facesSeperator}	
	Let $G=(V,E,w,K)$ be a weighted terminated planar graph. Suppose that $\gamma(G,K)=\gamma$.
	Then there are two shortest paths $P_1,P_2$ in $G$, such that for every connected component $C$ in $G\setminus\{P_1\cup P_2\}$ it holds $\gamma(G[C],K\cap C)\le\frac23\gamma+1$.
\end{restatable} 
\begin{proof}
	A planar cycle separator theorem for vertices have already became folklore \cite{Mil86,T04}. Specifically, given a weight function $\omega:V\rightarrow \R_{+}$ over the vertices, and a root vertex $v\in V$, one can efficiently find a cycle $S$, that consists of two shortest paths rooted at $v$, such that the total weight of the vertices in each connected component of $G[V\setminus S]$ is at most $\frac23\sum_{v\in V}\omega(v)$. 
	
	We start by defining a weight function $\omega$. Let $\mathcal{F}$ be a face cover of size $\gamma$. For every face $F\in\mathcal{F}$, let $v_F\in F$ be  an arbitrary vertex (not necessarily unique).
	Initially the weight of all the vertices is $0$. For every $F\in\mathcal{F}$, add a single unit of weight to $v_F$. Note that the total weight of all the vertices is $\gamma$, while for every $F\in\mathcal{F}$, the total weight of the vertices in $F$ is at least $1$. See \Cref{fig:seprator} for an illustration.
	
	Let $v$ be an arbitrary vertex on the outer face. We use the planar cycle separator theorem w.r.t. the weight function $\omega$ and the root vertex $v$. As a result, we get a pair of shortest paths $P_1,P_2$ rooted in $v$. Let $C$ be a connected component in $G\setminus\{P_1\cup P_2\}$.
	The total weight of all the vertices in $C$ is bounded by $\frac23\gamma$. 
	Consider the drawing of $C$ obtained by removing all other vertices from the drawing of $G$.
	Next we define a face cover $\mathcal{F}_C$. For every $F\in \mathcal{F}$, if $v_F\in C$ then add $F$ to $\mathcal{F}_C$ (or the new face containing the remainder of $F$). Additionally, add the outer face in the drawing of $C$ to $\mathcal{F}_C$.
	
	\begin{figure}[t]
		\centering{\includegraphics[scale=0.85]{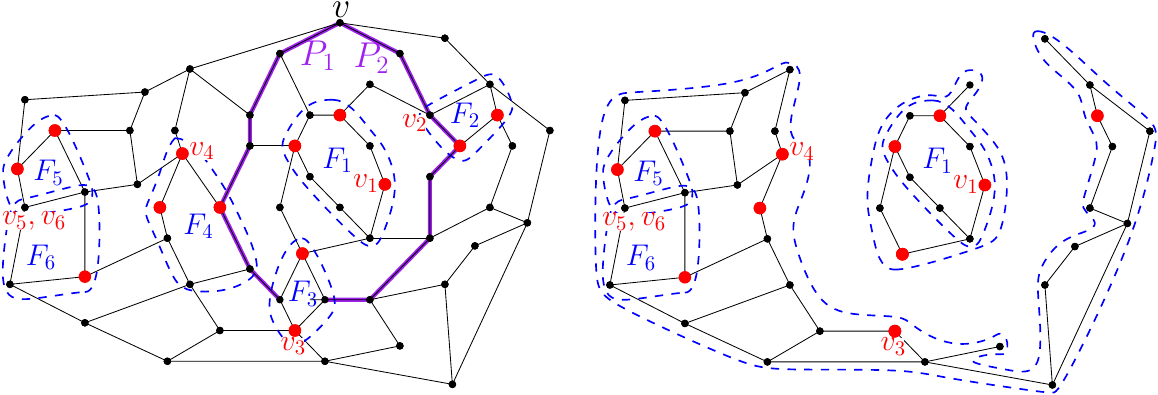}} 
		\caption{\label{fig:seprator}\small \it 
			On the left side displayed a graph $G$. The terminals are colored in red. The face cover consist of the faces $F_1,\dots,F_6$, surrounded by blue dashed lines. For each face $F_i$ let $v_{F_i}$ (denoted $v_i$) be an arbitrary vertex on $F_i$. Define a weight function $\omega$ by adding a unit of weight to every $v_i$. In the illustration $v_1,v_2,v_3,v_4$ have weight $1$, $v_5$ has weight $2$, while all other vertices have weight $0$.
			The separator consists of shortest paths $P_1,P_2$ colored purple.\quad
			On the right side we display the graph after removing all separator vertices. In each connected component $C$, a new face cover is defined  by taking the outer face and adding a single face for every $v_i\in C$.
		}
	\end{figure}
	
	It is straightforward that $|\mathcal{F}_C|\le \frac23\gamma+1$. We argue that $\mathcal{F}_C$ is a face cover for $K\cap C$. Indeed, let $u\in K\cap C$. Let $F\in \mathcal{F}$ be some face s.t. $u\in F$. If $v_F\in C$, then $F$ (or its remainder) is in $\mathcal{F}_C$, and therefore $u$ is covered.
	Otherwise, $v_F\notin C$. Therefore $v_F$ and $u$ were separated by the deletion of $P_1,P_2$. Necessarily some vertex of $F$ belongs to $P_1\cup P_2$. We conclude that $u$ is now part of the outer face, and therefore covered.
\end{proof}

If the size of the face cover is $2$, we can reduce this size to $1$ by removing a single shortest path containing vertices from both faces, the remaining graph will be O-S.
Similarly, if the size of the face cover is $3$ we can reduce to $1$ by removing a pair of shortest paths.
We can invoke \Cref{thm:facesSeperator} repeatedly in order to hierarchically partition $G$, reducing the size of the face cover in each iteration. After $O(\log \gamma)$ iterations, the size of the face cover in each connected component will be at most $1$. 
\begin{corollary}\label{cor:PlanarPSPD}
	Let $G=(V,E,w,K)$ be a weighted terminated planar  graph such that $\gamma(G,K)=\gamma$. Then there is an \PSPD $\left\{ \mathcal{X},\mathcal{P}\right\} $ of depth $O(\log \gamma)$ and remainder $\left\{ \mathcal{C},\mathcal{B}\right\} $, such that for every cluster $C\in \mathcal{C}$, all the terminals in $C$ lie on a single face (i.e. $\forall C\in\mathcal{C},~~\gamma(C,K\cap C)\le1$).
\end{corollary}

\section{Embedding Parametrized by Face Cover: Proof of \Cref{thm:main} }\label{sec:FaceCoverEmbedding}
We restate the theorem for convenience.
\MainEmbedding*
\begin{proof}
	Given the \PSPD $\left\{ \mathcal{X},\mathcal{P}\right\}$ from \Cref{cor:PlanarPSPD}, with remainder $\left\{ \mathcal{C},\mathcal{B}\right\}$, we are ready to define the embedding of \Cref{thm:main}.
	Let $f_{\PSPD}$ be the embedding of $G$ into $\ell_1$ from \Cref{thm:EmbPSPD}, restricted to $K$.
	For every cluster $C\in\mathcal{C}$ which contains a terminal, let $\mathcal{B}_C=V\setminus C$ and $F_C$ be  the single face in the face cover of $C$. Let $f_C$ be the embedding from \Cref{lem:BoundedOStoL1} with parameters $F_C,C,\mathcal{B}_C$. Let $\hat{f}_C$ be the embedding $f_C$ restricted to $K\cap C$, and extended to $K$ by sending every $v\in K\setminus C$ to $\vec{0}$. 
	The final embedding $f$ will be a concatenation of $f_{\PSPD}$ with $\hat{f}_C$ for all $C\in\mathcal{C}$.
	\paragraph{Expansion} Consider a pair of vertices $v,u\in V$. By \Cref{thm:EmbPSPD}, $\left\Vert f_{\PSPD}(v)-f_{\PSPD}(u)\right\Vert _{1}=O\left(\sqrt{\log \gamma}\right)\cdot d_{G}(v,u)$. On the other hand, for every $C\in \mathcal{C}$, using \Cref{lem:BoundedOStoL1}, if $u,v\in C$ then $\left\Vert f_C(v)-f_C(u)\right\Vert _{1}=O\left(1\right)\cdot d_{G}(v,u)$. Otherwise if $v\in C,u\notin C$, then $\left\Vert f_{C}(v)-f_{C}(u)\right\Vert_1 =d_{G}(v,\mathcal{B})\le d_{G}(v,u)$
	(similarly for $v\notin C,u\in C$). As each vertex is nonzero only in a single function $f_C$, the $O(\sqrt{\log \gamma})$ bound on the expansion follows.
	
	\sloppy\paragraph{Contraction} Consider a pair of terminal vertices $v,u$. If either $u,v$ are separated by $\mathcal{C}$ or $\min\left\{ d_{G}(v,\mathcal{B}),d_{G}(u,\mathcal{B})\right\} \le \frac{d_G(u,v)}{12}$, it holds that $\left\Vert f_{\PSPD}(v)-f_{\PSPD}(u)\right\Vert _{1}\ge d_G(u,v)$ and we are done.
	Otherwise, there must exist a cluster $C\in\mathcal{C}$ such that $u,v\in F_C$ and $\min\left\{ d_{G}(v,\mathcal{B}_C),d_{G}(u,\mathcal{B}_C)\right\} \ge \frac{d_G(u,v)}{12}$. By \Cref{lem:BoundedOStoL1} $\left\Vert \hat{f}_{C}(v)-\hat{f}_{C}(u)\right\Vert _{1}=\left\Vert f_{C}(v)-f_{C}(u)\right\Vert _{1}\ge \frac{d_G(u,v)}{12\cdot 12}$.
	
	\Cref{thm:main} now follows. Bellow we discuss the implementation details of the embedding.
	
	\paragraph{Polynomial Implementation}
	Given a planar graph with a drawing in the plane, using the PTAS of Frederickson \cite{Fre91} we can find a face cover of size $2\cdot \gamma(G,K)$ in linear time (see \Cref{subsec:related}).
	Note that using a cover of size $2\cdot \gamma(G,K)$ instead of $\gamma(G,K)$ is insignificant for our $O(\sqrt{\log(\gamma(G,K))}\ )$ upper bound.
	Next, construct a \PSPD for this face cover using cycle separators. Since we construct at most $n$ separators, the construction of the \PSPD also takes polynomial time.
	Given a \PSPD, the embedding of \cite{AFGN22} is efficiently computed.
	
	After the creation of $f_{\PSPD}$ we are left with a remainder $\{\mathcal{C},\mathcal{B}\}$. For every $C\in \mathcal{C}$ we start by computing uniformly truncated embeddings (see \Cref{remark:time}).
	Given an O-S graph with $k$ terminals, there are at most $2k$ truncation scales (as each terminal participates in two scales only). Thus in polynomial time we can compute the embedding for all truncation scales, and thus compute the non-uniformly truncated embedding.
\end{proof}

\section{The Lipschitz Extension Problem (and proof of \Cref{cor:main})}\label{sec:LipschitzExtension}
Given an embedding $f:X\rightarrow Y$ from a metric space $(X,d_X)$ to a metric space $(Y,d_Y)$ its expansion $\|f\|_{\Lip}=\sup_{x,y\in X}\frac{d_Y(f(x),f(y))}{d_X(x,y)}$ is called the Lipschitz parameter of the function.
In the Lipschitz extension problem, we are given a map $f:Z\rightarrow Y$ from a subset $Z$ of $X$. The goal is to extend $f$ to a function $\tilde{f}$ over the entire space $X$, while minimizing $\|\tilde{f}\|_{\Lip}$ as a function of $\|f\|_{\Lip}$.

A partition $\cP$ of a metric space $(X,d_X)$ is said to be $\Delta$-bounded if every cluster has diameter at most $\Delta$. For a point $x$, we will denote by $P(x)$ the cluster in $\cP$ containing $x$.
A stochastic decomposition of a metric space is a distribution of partitions of bounded diameter such that nearby points are likely to be clustered together (often called low diameter decompositions or LDD's).
\begin{definition}[Stochastic Decomposition]\label{def:Decompostion}
	A distribution $\mathcal{D}$ over partitions of a metric space $(X,d_X)$ is  $(\beta,\Delta)$-stochastic decomposition if every $\mathcal{P}\in\supp(\mathcal{D})$ is  $\Delta$-bounded and for every pair  $x,y\in X$,
	$$\Pr[P(v)\ne P(u)] \le \beta\cdot \frac{d_X(x,y)}{\Delta}~.$$
	%
	%
	We say that $(X,d_X)$ admits a $\beta$-stochastic decomposition scheme if for every parameter $\Delta>0$ it admits a $(\beta,\Delta)$-stochastic decomposition.
\end{definition}
It is known that general $n$-point metric space admit $O(\log n)$-stochastic decomposition scheme \cite{B96}, metric space with doubling dimension $d$ admit $O(d)$-stochastic decomposition scheme \cite{GKL03,Fil19padded}, $d$-dimensional Euclidean space admits $O(\sqrt{d})$-stochastic decomposition scheme (this follows from \cite{CCGGP98,AI06}, see also \cite{Fil23}), $K_r$ minor free graphs admit  $O(r)$-stochastic decomposition scheme \cite{KPR93,FT03,AGGNT19,Fil19padded}, and graphs with treewidth $\tw$ admit $O(\log\tw)$-stochastic decomposition scheme \cite{FFIKLMZ24}.

It follows from Lee and Naor \cite{LN05} that give function $f$ from a subset of a metric space $(X,d_X)$ that admits a $O(\beta)$-stochastic decomposition scheme, into a closed convex subset of a Banach space\footnote{For the reader not familiar with Banach spaces, just replace it with $\ell_1$.}, $f$ can be extended to the entire space $X$ with Lipschitz constant  $\|\tilde{f}\|_{\Lip}=O(\beta)\cdot\|f\|_{\Lip}$:
\begin{theorem}\label{thm:LipschitzExtension}
	Let $(X,d_X)$ be a finite metric space that admits a $\beta$-stochastic decomposition scheme, $(Z,\|\cdot\|)$ a Banach space, and $f:K\rightarrow Y$ a function from subset $K\subseteq X$ to a convex subset $Y\subseteq Z$. 
	Then there is an extension $\tilde{f}:X\rightarrow Y$ with Lipschitz constant  $\|\tilde{f}\|_{\Lip}=O(\beta)\cdot\|f\|_{\Lip}$.\\
	Further, if we can efficiently sample from the stochastic decompositions, then the extension $\tilde{f}$ can be computed efficiently.
\end{theorem}
We prove \Cref{thm:LipschitzExtension} explicitly in \Cref{sec:LipschitzProof}. Our proof particularly proves that such an extension can be computed efficiently. 
Using \Cref{thm:main} and than \Cref{thm:LipschitzExtension} on top of any of the known stochastic decompositions of planar graphs (e.g. \cite{KPR93}), \Cref{cor:main} follows.
%
We restate it for convenience.
\EmbeddingAllGraph*
\begin{proof}
	Using \Cref{thm:main}, let $f:K\rightarrow\ell_1$ be an embedding with distortion $O(\sqrt{\log\gamma(G,K)})$. 
	Specifically, for every $u,v\in K$, $\|f(u)-f(v)\|_1\le d_X(u,v)\le O(\sqrt{\log\gamma(G,K)})\cdot \|f(u)-f(v)\|_1$.
	According to \cite{KPR93} planar graphs admit efficiently computable $O(1)$-stochastic decomposition scheme.
	Using \Cref{thm:LipschitzExtension}, we obtain an extension $\tilde{f}:V\rightarrow\ell_1$ such that for every $u,v\in X$, $d_X(u,v)\le O(\sqrt{\log\gamma(G,K)})\cdot \|f(u)-f(v)\|_1$.
	The total computation time is polynomial.
\end{proof}

\subsection{Proof of \Cref{thm:LipschitzExtension}}\label{sec:LipschitzProof}
%



We begin be defining the extension function $\tilde{f}$.
Given a subset $S\subseteq X$, let $S_K\subseteq K$ be the corresponding subset in $K$ defined as follows:
If $S\cap K\ne \emptyset$, let $S_K=S\cap K$. Otherwise ($S\cap K\ne \emptyset$), let $S_K$ be the closest terminal to $S$. That is the terminal $x\in K$ minimizing $d_X(S,x)$ (breaking ties arbitrarily).

Given a partition $\cP$ of $X$, we define an extension $f_{\cP}$ of $f$ as follows, for every $x\in K$, $f_{\cP}(x)=x$. For $x\notin K$ it will be simply the barycenter (center of mass) of the corresponding subset of terminals $P(x)_K$. That is 
\[
f_{\cP}(x)=\frac{1}{|P(x)_K|}\cdot\sum_{y\in P(x)_K}f(y)~,
\]
Note that as $Y$ is convex, $f_{\cP}(x)\in Y$. 

Fix a scale $2^{t}$, and let ${\cal D}_{t}$ be a $(\beta,\beta\cdot2^t)$-stochastic decomposition $X$ (that is the bound on the diameter is $\beta\cdot2^t$). We define the following extension:
\[
f_{t}(x)=\E_{\cP\sim{\cal D}_{t}}\left[f_{\cP}(x)\right]=\sum_{\cP\sim{\cal D}_{t}}\Pr[\cP]\cdot f_{\cP}(y)~,
\]

For a point $x$, let $t_{x}\in\Z$ be such that $d_{X}(x,K)\in[2^{t_{x}},2^{t_{x}+1})$. Set $\lambda_x=\frac{d_{X}(x,K)-2^{t_x}}{2^{t_x}}$.  Note that $0\le
\lambda_x<1$.
We define extension on the entire metric space as follows:
\begin{gather*}
	\tilde{f}(x)\doteq (1-\lambda_{x})\cdot f_{t}(x) +\lambda_{x}\cdot f_{t+1}(x)~.
	\label{eq:extension}
\end{gather*}
Clearly, as each $f_t$ is extension of $f$, so does $\tilde{f}$.
Note that $f(x)$ is a convex combination of the maps of the terminals in $K_x=\left\{P(x)_K\mid\cP\in\supp(\cD_{t_x})\cup\supp(\cD_{t_x+1})\right\}$. The terminals $K_x$ are at distance at most $(2\beta+2)\cdot 2^{t_x}\le (2\beta+2)\cdot d_X(x,K)$
from $x$. Indeed, for every partition $\cP\in\supp(\cD_i)\cup\supp(\cD_{i+1})$, if $P(x)\cap K\ne\emptyset$, then all the terminals in $P(x)_K$ belong to a cluster containing $x$ of diameter at most $\beta\cdot 2^{t_x+1}=2\beta\cdot 2^{t_x}$. Otherwise ($P(x)\cap K=\emptyset$), $P(x)_K$ is a singleton terminal $z\in K$, which is the closest terminal to $P(x)$ at distance at most $d_X(P(x),K)\le d_X(x,K)$ from $P(x)$. Let $z'\in P(x)$ be the closest point in $P(x)$ to a point $z\in K$ (that is $d_X(z,z')=d_X(P(x),K)$). Then it holds that 
$d_{X}(x,z)\le d_{X}(x,z')+d_{X}(z',z)\le \beta\cdot\Delta_{t_{x}+1}+d_{X}(x,K)<(2\beta+2)\cdot 2^{t_x}$.

We argue that $\tilde{f}$ has Lipschitz constant $\|\tilde{f}\|_{\Lip}\le O(\beta)\cdot \|f\|_{\Lip}$.  
Consider a pair of points $x,y\in X$.
If both $x,y\in K$ then there is nothing to prove.
Consider next the case where $x\in X\setminus K$, and $y\in K$. 
Let $\tilde{f}(x)=\sum_{z\in K_{x}}\alpha_{z}\cdot f(z)$ be the convex combination of $\tilde{f}(x)$ (that is 
$\{\alpha_{z}\}_{z\in K_{x}}\ge 0$ and $\sum_{z\in K_{x}}\alpha_{z}=1$).
It holds that 
\begin{align*}
	\left\Vert \tilde{f}(x)-\tilde{f}(y)\right\Vert  &  \le\sum_{z\in K_{x}}\alpha_{z}\cdot\left\Vert f(z)-f(y)\right\Vert \\
	& \le\sum_{z\in K_{x}}\alpha_{z}\cdot\|f\|_{\Lip}\cdot d_{X}(z,y)\\
	& \le\|f\|_{\Lip}\cdot\sum_{z\in K_{x}}\alpha_{z}\cdot\left(d_{X}(z,x)+d_{X}(x,y)\right)\\
	& \le\|f\|_{\Lip}\cdot\left(2\beta+3\right)\cdot d_{X}(x,y)~,
\end{align*}
where in the last inequality we used that $d_{X}(z,x)\le (2\beta+2)\cdot d_X(x,K)\le (2\beta+2)\cdot  d_X(x,y)$.

Next, we move to the most interesting case where $x,y\in X\setminus K$. 
Suppose w.l.o.g. that $d_X(y,K)\le d_X(x,K)$ (implying $t_y\le t_x$). 
The following claim will be useful.
\begin{claim}\label{clm:arbitraryt}
	For every $t\in\{t_x,t_x+1\}$ and $t'\in\{t_y,t_y+1\}$, \\\phantom{.}\hfill$\left\Vert f_{t}(x)-f_{t'}(y)\right\Vert \le\|f\|_{\Lip}\cdot\left((8\beta+8)\cdot2^{t_{x}}+d_{X}(x,y)\right)$.
\end{claim}
\begin{proof}
	$f_t(x)=\sum_{z\in K_{x}}\alpha_{z}\cdot f(z)$ and $f_t(y)=\sum_{w\in K_{y}}\gamma_{w}\cdot f(w)$ be the convex combinations of $f_t(x)$ and $f_t(y)$.	
	As $\sum_{z\in K_{x}}\alpha_{z}=\sum_{w\in K_{y}}\gamma_{z}=1$, it holds that 
	\begin{align}
		\left\Vert f_{t}(x)-f_{t'}(y)\right\Vert  & =\left\Vert \sum_{z\in K_{x}}\alpha_{z}\cdot f(z)-\sum_{w\in K_{y}}\gamma_{z}\cdot f(w)\right\Vert \nonumber\\
		& \le\max_{z\in K_{x},w\in K_{y}}\left\Vert f(z)-f(w)\right\Vert\nonumber \\
			& \le\max_{z\in K_{x},w\in K_{y}}\|f\|_{\Lip}\cdot d_{X}(z,w)\nonumber\\
		& \le\|f\|_{\Lip}\cdot\max_{z\in K_{x},w\in K_{y}}\left(d_{X}(z,x)+d_{X}(x,y)+d_{X}(y,w)\right)\nonumber\\
		& \le\|f\|_{\Lip}\cdot\left((2\beta+2)\cdot d_{X}(x,K)+d_{X}(x,y)+(2\beta+2)\cdot d_{X}(y,K)\right)\label{eq:tt'}\\
		& \le\|f\|_{\Lip}\cdot\left((8\beta+8)\cdot2^{t_{x}}+d_{X}(x,y)\right)~.\nonumber
	\end{align}
\end{proof}

Next, consider the case where $d_X(x,K)\ge2\cdot d_X(y,K)$. It holds that
$d_{X}(x,y)\ge d_{X}(x,K)-d_{X}(y,K)\ge d_{X}(y,K)$. Hence using inequality \ref{eq:tt'} and the fact that $\tilde{f}(x)$ is a convex combination of $f_{t_x}(x),f_{t_x+1}(x)$ (resp. $\tilde{f}(y)$ is a convex combination of $f_{t_y}(y),f_{t_y+1}(y)$), we have that
\begin{align*}
	\left\Vert \tilde{f}(x)-\tilde{f}(y)\right\Vert  & \le\|f\|_{\Lip}\cdot\left((2\beta+2)\cdot d_{X}(x,K)+d_{X}(x,y)+(2\beta+2)\cdot d_{X}(y,K)\right)\\
	& \le\|f\|_{\Lip}\cdot\left((2\beta+3)\cdot d_{X}(x,y)+(4\beta+4)\cdot d_{X}(y,K)\right)\\
	& \le\|f\|_{\Lip}\cdot\left(6\beta+7\right)\cdot d_{X}(x,y)~,
\end{align*}
here in the second inequality we used that $d_X(x,K)\le d_X(x,y)+d_X(y,K)$.
We will thus assume that $d_X(y,K)\le d_X(x,K)<2\cdot d_X(y,K)$. In particular, it holds that $t_y\le t_x\le t_y+1$.

Suppose next that $d_X(x,y)\ge d_X(y,K)$. 
In particular, $d_X(x,y)\ge 2^{t_y}\ge\frac12\cdot 2^{t_x}$.
Using \Cref{clm:arbitraryt}, and the fact that  $\tilde{f}(x)$ is a convex combination of $f_{t_x}(x),f_{t_x+1}(x)$ (resp. $\tilde{f}(y)$ is a convex combination of $f_{t_y}(y),f_{t_y+1}(y)$), we have that 
\begin{align*}
	\left\Vert \tilde{f}(x)-\tilde{f}(y)\right\Vert  & \le\|f\|_{\Lip}\cdot\left((8\beta+8)\cdot2^{t_{x}}+d_{X}(x,y)\right)\\
	& \le\|f\|_{\Lip}\cdot\left((8\beta+8)\cdot2+1\right)\cdot d_{X}(x,y)~.
\end{align*}
Hence we can assume that $d_X(x,y)< d_X(y,K)$.
Next, using this assumption, we prove the following:%
\begin{claim}\label{clm:ft}
	For every $t\in \{t_x,t_x+1\}\cap \{t_y,t_y+1\}$, $\left\Vert f_{t}(x)-f_{t}(y)\right\Vert \le \|f\|_{\Lip}\cdot\left(8\beta+10\right)\cdot d_{X}(x,y)$.
\end{claim}
\begin{proof}
	Consider the distribution $\cD_t$, and let $\cP\in \supp(\cD_t)$.
	If $P(x)=P(y)$, then $f_{\cP}(x)=f_{\cP}(y)$ and thus $\left\Vert f_{\cP}(x)-f_{\cP}(y)\right\Vert=0$. Otherwise, $P(x)\ne P(y)$. Here using the same reasoning as in \Cref{clm:arbitraryt}, 
	\begin{align*}
		\left\Vert f_{\cP}(x)-f_{\cP}(y)\right\Vert  & \le\max_{z\in K_{x},w\in K_{y}}\|f\|_{\Lip}\cdot d_{X}(z,w)\\
		& \le\|f\|_{\Lip}\cdot\left((8\beta+8)\cdot2^{t_{x}}+d_{X}(x,y)\right)\\
		& \le\|f\|_{\Lip}\cdot(8\beta+10)\cdot2^{t_{x}}~.
	\end{align*}

%
	Using the fact that $\cD_t$ is a $(\beta,\beta\cdot2^t)$-stochastic decomposition, we conclude
	\begin{align*}
		\left\Vert f_{t}(x)-f_{t}(y)\right\Vert  & \le\sum_{\cP\sim{\cal D}_{t}}\Pr[\cP]\cdot\left\Vert f_{\cP}(x)-f_{\cP}(y)\right\Vert \\
		& \le\Pr_{\cP\sim{\cal D}_{t}}[P(x)\ne P(y)]\cdot\|f\|_{\Lip}\cdot\left(8\beta+10\right)\cdot2^{t_{x}}\\
		& \le\beta\cdot\frac{d_{X}(x,y)}{\beta\cdot2^{t}}\cdot\|f\|_{\Lip}\cdot\left(8\beta+10\right)\cdot2^{t_{x}}\\
		& \le\|f\|_{\Lip}\cdot\left(8\beta+10\right)\cdot d_{X}(x,y)~.
	\end{align*}
\end{proof}

Denote by $p_t$ the scaling factor of $x$ in $f_{t}(x)$. That is, $p_{t_x+1}=\lambda_x$, $p_{t_x}=1-\lambda_x$, and $p_t=0$ for $t\notin\{t_x,t_{x}+1\}$. Similarly, define $q_t$ for $y$.
The rest of the proof is by case analysis.

	\paragraph{Case 1: $t_{x}=t_{y}$.} Note that as we assume $d_X(x,K)\ge d_X(y,K)$, $\lambda_{x}\ge\lambda_{y}$,
	and thus $p_{t_{x}}\le q_{t_{x}}$, $p_{t_{x}+1}\ge q_{t_{x}+1}$. Using  \Cref{clm:arbitraryt} and \Cref{clm:ft},
	
	\begin{align*}
		& \left\Vert \tilde{f}_{t}(x)-\tilde{f}_{t}(y)\right\Vert \\
		& =\left\Vert p_{t_{x}}\cdot f_{t_{x}}(x)+p_{t_{x}+1}\cdot f_{t_{x}+1}(x)-q_{t_{x}}\cdot f_{t_{x}}(y)-q_{t_{x}+1}\cdot f_{t_{x}+1}(y)\right\Vert \\
		& \le p_{t_{x}}\cdot\left\Vert f_{t_{x}}(x)-f_{t_{x}}(y)\right\Vert +q_{t_{x}+1}\cdot\left\Vert f_{t_{x}+1}(x)-f_{t_{x}+1}(y)\right\Vert +(1-p_{t_{x}}-q_{t_{x}+1})\cdot\left\Vert f_{t_{x}+1}(x)-f_{t_{x}}(y)\right\Vert \\
		& \le\|f\|_{\Lip}\cdot O(\beta)\cdot d_{X}(x,y)+(1-p_{t_{x}}-q_{t_{x}+1})\cdot\|f\|_{\Lip}\cdot\left(O(\beta)\cdot2^{t_{x}}+d_{X}(x,y)\right)\\
		& \stackrel{(*)}{\le}\|f\|_{\Lip}\cdot O(\beta)\cdot d_{X}(x,y)+\frac{d_{X}(x,y)}{2^{t_{x}}}\cdot\|f\|_{\Lip}\cdot O(\beta)\cdot2^{t_{x}}=\|f\|_{\Lip}\cdot O(\beta)\cdot d_{X}(x,y)~,
	\end{align*}
	where the inequality $^{(*)}$ holds as 
	\[
	1-p_{t_{x}}-q_{t_{x}+1}=1-(1-\lambda_{x})-\lambda_{y}=\frac{d_{X}(x,K)-2^{t_{x}}}{2^{t_{x}}}-\frac{d_{X}(y,K)-2^{t_{x}}}{2^{t_{x}}}\le\frac{d_{X}(x,y)}{2^{t_{x}}}~.
	\]
	
	\paragraph{Case 2: $t_{x}=t_{y}+1$ and $p_{t_x}\ge q_{t_x}$.} Using  \Cref{clm:arbitraryt} and \Cref{clm:ft},
\begin{align*}
	& \left\Vert \tilde{f}_{t}(x)-\tilde{f}_{t}(y)\right\Vert \\
	& =\left\Vert p_{t_{x}}\cdot f_{t_{x}}(x)+p_{t_{x}+1}\cdot f_{t_{x}+1}(x)-q_{t_{x}-1}\cdot f_{t_{x}-1}(y)-q_{t_{x}}\cdot f_{t_{x}}(y)\right\Vert \\
	& \le q_{t_{x}}\cdot\left\Vert f_{t_{x}}(x)-f_{t_{x}}(y)\right\Vert +(p_{t_{x}}-q_{t_{x}})\cdot\left\Vert f_{t_{x}}(x)-f_{t_{x}-1}(y)\right\Vert +p_{t_{x}+1}\cdot\left\Vert f_{t_{x}+1}(x)-f_{t_{x}-1}(y)\right\Vert \\
	& \le\|f\|_{\Lip}\cdot O(\beta)\cdot d_{X}(x,y)+(1-q_{t_{x}})\cdot\|f\|_{\Lip}\cdot\left(O(\beta)\cdot2^{t_{x}}+d_{X}(x,y)\right)\\
	& \stackrel{(**)}{\le}\|f\|_{\Lip}\cdot O(\beta)\cdot d_{X}(x,y)+\frac{d_{X}(x,y)}{2^{t_{x}-1}}\cdot\|f\|_{\Lip}\cdot O(\beta)\cdot2^{t_{x}}=\|f\|_{\Lip}\cdot O(\beta)\cdot d_{X}(x,y)
\end{align*}
where the inequality $^{(**)}$ holds as 
\[
1-q_{t_{x}}=1-\lambda_{y}=1-\frac{d_{X}(y,K)-2^{t_{y}}}{2^{t_{y}}}=\frac{2^{t_{x}}-d_{X}(y,K)}{2^{t_{x}-1}}\le\frac{d_{X}(x,K)-d_{X}(y,K)}{2^{t_{x}-1}}\le\frac{d_{X}(x,y)}{2^{t_{x}-1}}~.
\]

	\paragraph{Case 3: $t_{x}=t_{y}+1$ and $p_{t_x}< q_{t_x}$.} Using  \Cref{clm:arbitraryt} and \Cref{clm:ft},
\begin{align*}
	& \left\Vert \tilde{f}_{t}(x)-\tilde{f}_{t}(y)\right\Vert \\
	& =\left\Vert p_{t_{x}}\cdot f_{t_{x}}(x)+p_{t_{x}+1}\cdot f_{t_{x}+1}(x)-q_{t_{x}-1}\cdot f_{t_{x}-1}(y)-q_{t_{x}}\cdot f_{t_{x}}(y)\right\Vert \\
	& \le p_{t_{x}}\cdot\left\Vert f_{t_{x}}(x)-f_{t_{x}}(y)\right\Vert +(q_{t_{x}}-p_{t_{x}})\cdot\left\Vert f_{t_{x}+1}(x)-f_{t_{x}}(y)\right\Vert +q_{t_{x}-1}\cdot\left\Vert f_{t_{x}+1}(x)-f_{t_{x}-1}(y)\right\Vert \\
	& \le\|f\|_{\Lip}\cdot O(\beta)\cdot d_{X}(x,y)+(1-p_{t_{x}})\cdot\|f\|_{\Lip}\cdot\left(O(\beta)\cdot2^{t_{x}}+d_{X}(x,y)\right)\\
	& \stackrel{(***)}{\le}\|f\|_{\Lip}\cdot O(\beta)\cdot d_{X}(x,y)+\frac{d_{X}(x,y)}{2^{t_{x}}}\cdot\|f\|_{\Lip}\cdot O(\beta)\cdot2^{t_{x}}=\|f\|_{\Lip}\cdot O(\beta)\cdot d_{X}(x,y)
\end{align*}
where the inequality $^{(***)}$ holds as 
\[
1-p_{t_{x}}=1-(1-\lambda_{x})=\frac{d_{X}(x,K)-2^{t_{x}}}{2^{t_{x}}}<\frac{d_{X}(x,K)-d_{X}(y,K)}{2^{t_{x}}}\le\frac{d_{X}(x,y)}{2^{t_{x}}}~.
\]

\subsubsection{Computational Aspects}
It remains to prove that if we can sample from the stochastic decompositions efficiently, than we can compute the extension efficiently. 
In \Cref{lem:efficinetDecomposition} below we show that given a $(\beta,\Delta)$-stochastic decomposition $\cD$ from which we can sample efficiently, one can compute a $(\beta,\Delta)$-stochastic decomposition $\cD'$ with polynomial size support.
To compute the extension, we will simply use these decompositions with polynomial size support.
Then it is clear that the extension defined in the beginning of the proof can be computed efficiently.

\begin{lemma}\label{lem:efficinetDecomposition}
	Consider an $n$-point metric space $(X,d_X)$ with a $(\beta,\Delta)$-stochastic decomposition $\cD$ such that we can sample from $\cD$ in a polynomial time.
	Then there is an efficiently commutable  $(2\beta,\Delta)$-stochastic decomposition $\cD'$ with support size $\poly(n)$.
\end{lemma}
\begin{proof}
	Fix $\delta=\frac18$, $\eps=\delta\cdot{n\choose 2}^{-1}=\frac18\cdot{n\choose 2}^{-1}$, and $N=\frac{2}{\eps^{2}}\cdot\ln n^{3}=\tilde{O}(n^{4})$.
	We sample $N$ partitions $\cP_1,\cP_2,\dots,\cP_N$ from $\cD$ independently.
	We will denote by $P_i(x)$ the cluster of $x$ in $\cP_i$.
	A pair of points $x,y\in X$ are called neighbors if $d_X(x,y)\le \eps\cdot\frac\Delta\beta$.
	We say that a partition $\cP_i$ is bad if any pair of neighbors is separated. That is there are neighbors $x,y$ such that $P_i(x)\ne P_i(y)$.  
	Denote by $B\subseteq[1,N]$ the indices of bad partitions.
	Using union bound, the probability that a sampled partition is bad is bounded by  ${n\choose 2}\cdot \beta\cdot\frac{\eps\cdot\frac\Delta\beta}{\Delta}=\eps\cdot{n\choose 2}=\delta$. Hence the expected number of bad partitions is at most $|B|\le\delta\cdot N$.
	Let $\Psi_{\rm bad}$ be the event that there are at least $|B|>2\delta\cdot N$ bad partitions. Using Chernoff inequality\footnote{See e.g. Theorem 7.3.5  \href{https://sarielhp.org/misc/blog/15/09/03/chernoff.pdf}{here}.\label{foot:Chernoff}}
	\[\Pr[\Psi_{\rm bad}]\le \Pr\left[|B|-\E[|B|]\ge\delta\cdot N\right]\le e^{-\frac{2\cdot\delta^2\cdot N^2}{N}}=e^{-2\cdot\delta^2\cdot N}~.\]
	
	Consider a pair $x,y\in X$ of non-neighboring vertices. 
	Let $I_{x,y}\subseteq [1,N]$ be the indices of the partitions where $x$ and $y$ been separated. In expectation, $\E[I_{x,y}]\le N\cdot \beta\cdot\frac{d_X(x,y)}{\Delta}$. 
	Let $\Psi_{x,y}$ be the event that there are at least $|I_{x,y}|>\frac{3}{2}N\cdot\beta\cdot\frac{d_{X}(x,y)}{\Delta}\ge\E[I_{x,y}]+\frac{1}{2}N\cdot\varepsilon$ partitions where $x$ and $y$ are separated.
	Using Chernoff inequality again:$^{\ref{foot:Chernoff}}$
	\[
	\Pr[\Psi_{x,y}]\le\Pr\left[|I_{x,y}|-\E[I_{x,y}]\ge\frac{\varepsilon}{2}\cdot N\right]\le e^{-\frac{2\cdot\frac{\varepsilon^{2}}{4}\cdot N^{2}}{N}}=e^{-\frac{\varepsilon^{2}}{2}\cdot N}~.
	\]
	By union bound, the probability that either of $\Psi_{\rm bad}$ or $\{\Psi_{x,y}\}_{x,y\in X}$ occurs is at most 
	\[
	e^{-2\cdot\delta^{2}\cdot N}+{n \choose 2}\cdot e^{-\frac{\varepsilon^{2}}{2}\cdot N}\le\frac{1}{2}~.
	\]
	We will assume that we sampled the $N$ partitions and assume that none of these events occur (otherwise re-sample).
	We define a stochastic-decomposition $\cD'$ as follows: simply sample uniformly at random a partition among $\cP_1,\cP_2,\dots,\cP_N$ that did not separated a neighboring pair. 
	Clearly  $\cD'$ has support size at most $N=\tilde{O}(n^{4})$ and is $\Delta$-bounded. It remains to show that every pair $x,y$ is separated with probability at most $2\beta\cdot\frac{d_X(x,y)}{\Delta}$.
	
	Fix a pair $x,y\in X$. If $x,y$ are neighbors, then as we removed all the partitions separating neighbors, $x,y$ are separated in $\cD'$ with probability $0$.
	Next, suppose that $x,y$ are not neighbors.
	As we sample a partition uniformly at random, the probability that $x,y$ are separated is at most
	\[
	\Pr_{\cP_{i}\sim\cD'}\left[P_{i}(x)\ne P_{i}(y)\right]\le\frac{|I_{x,y}|}{N-|B|}\le\frac{\frac{3}{2}N\cdot\beta\cdot\frac{d_{X}(x,y)}{\Delta}}{N-2\delta\cdot N}=\frac{\frac{3}{2}}{1-2\delta}\cdot\beta\cdot\frac{d_{X}(x,y)}{\Delta}\le2\beta\cdot\frac{d_{X}(x,y)}{\Delta}~,
	\]
	as required.
\end{proof}

\section*{Acknowledgments}
The author would like to thank Ofer Neiman for helpful discussions.

{\small
  \bibliographystyle{alphaurlinit}  
  \bibliography{bib-extended,art}

\newcommand{\etalchar}[1]{$^{#1}$}
\begin{thebibliography}{AGG{\etalchar{+}}19}

\bibitem[ABN11]{ABN11}
I.~Abraham, Y.~Bartal, and O.~Neiman.
\newblock Advances in metric embedding theory.
\newblock {\em Advances in Mathematics}, 228(6):3026 -- 3126,
\newblock 2011, \href {http://dx.doi.org/10.1016/j.aim.2011.08.003}
  {\path{doi:10.1016/j.aim.2011.08.003}}.

\bibitem[AFGN22]{AFGN22}
I.~Abraham, A.~Filtser, A.~Gupta, and O.~Neiman.
\newblock Metric embedding via shortest path decompositions.
\newblock {\em {SIAM} J. Comput.}, 51(2):290--314,
\newblock 2022, \href {http://dx.doi.org/10.1137/19M1296021}
  {\path{doi:10.1137/19M1296021}}.

\bibitem[AGG{\etalchar{+}}19]{AGGNT19}
I.~Abraham, C.~Gavoille, A.~Gupta, O.~Neiman, and K.~Talwar.
\newblock Cops, robbers, and threatening skeletons: Padded decomposition for
  minor-free graphs.
\newblock {\em {SIAM} J. Comput.}, 48(3):1120--1145,
\newblock 2019, \href {http://dx.doi.org/10.1137/17M1112406}
  {\path{doi:10.1137/17M1112406}}.

\bibitem[AI06]{AI06}
A.~Andoni and P.~Indyk.
\newblock Near-optimal hashing algorithms for approximate nearest neighbor in
  high dimensions.
\newblock In {\em 47th Annual {IEEE} Symposium on Foundations of Computer
  Science {(FOCS} 2006), 21-24 October 2006, Berkeley, California, USA,
  Proceedings}, pages 459--468. {IEEE} Computer Society,
\newblock 2006, \href {http://dx.doi.org/10.1109/FOCS.2006.49}
  {\path{doi:10.1109/FOCS.2006.49}}.

\bibitem[Bar96]{B96}
Y.~Bartal.
\newblock Probabilistic approximations of metric spaces and its algorithmic
  applications.
\newblock In {\em FOCS}, pages 184--193,
\newblock 1996.

\bibitem[Ber90]{Ber90}
M.~W. Bern.
\newblock Faster exact algorithms for steiner trees in planar networks.
\newblock {\em Networks}, 20(1):109--120,
\newblock 1990, \href {http://dx.doi.org/10.1002/net.3230200110}
  {\path{doi:10.1002/net.3230200110}}.

\bibitem[BFN19]{BFN19}
Y.~Bartal, A.~Filtser, and O.~Neiman.
\newblock On notions of distortion and an almost minimum spanning tree with
  constant average distortion.
\newblock {\em J. Comput. Syst. Sci.}, 105:116--129,
\newblock 2019, \href {http://dx.doi.org/10.1016/j.jcss.2019.04.006}
  {\path{doi:10.1016/j.jcss.2019.04.006}}.

\bibitem[BM88]{BM88}
D.~Bienstock and C.~L. Monma.
\newblock On the complexity of covering vertices by faces in a planar graph.
\newblock {\em {SIAM} J. Comput.}, 17(1):53--76,
\newblock 1988, \href {http://dx.doi.org/10.1137/0217004}
  {\path{doi:10.1137/0217004}}.

\bibitem[CCG{\etalchar{+}}98]{CCGGP98}
M.~Charikar, C.~Chekuri, A.~Goel, S.~Guha, and S.~A. Plotkin.
\newblock Approximating a finite metric by a small number of tree metrics.
\newblock In {\em 39th Annual Symposium on Foundations of Computer Science,
  {FOCS} '98, November 8-11, 1998, Palo Alto, California, {USA}}, pages
  379--388. {IEEE} Computer Society,
\newblock 1998, \href {http://dx.doi.org/10.1109/SFCS.1998.743488}
  {\path{doi:10.1109/SFCS.1998.743488}}.

\bibitem[CCL{\etalchar{+}}23]{CCLMST23}
H.~Chang, J.~Conroy, H.~Le, L.~Milenkovic, S.~Solomon, and C.~Than.
\newblock Covering planar metrics (and beyond): {O(1)} trees suffice.
\newblock {\em CoRR}, abs/2306.06215, 2023.
\newblock
\newblock To appear in FOCS 2023, \href {http://arxiv.org/abs/2306.06215}
  {\path{arXiv:2306.06215}}, \href
  {http://dx.doi.org/10.48550/ARXIV.2306.06215}
  {\path{doi:10.48550/ARXIV.2306.06215}}.

\bibitem[CFKL20]{CFKL20}
V.~Cohen{-}Addad, A.~Filtser, P.~N. Klein, and H.~Le.
\newblock On light spanners, low-treewidth embeddings and efficient traversing
  in minor-free graphs.
\newblock In S.~Irani, editor, {\em 61st {IEEE} Annual Symposium on Foundations
  of Computer Science, {FOCS} 2020, Durham, NC, USA, November 16-19, 2020},
  pages 589--600. {IEEE},
\newblock 2020, \href {http://dx.doi.org/10.1109/FOCS46700.2020.00061}
  {\path{doi:10.1109/FOCS46700.2020.00061}}.

\bibitem[CFW12]{CFW12}
A.~Chakrabarti, L.~Fleischer, and C.~Weibel.
\newblock When the cut condition is enough: a complete characterization for
  multiflow problems in series-parallel networks.
\newblock In H.~J. Karloff and T.~Pitassi, editors, {\em Proceedings of the
  44th Symposium on Theory of Computing Conference, {STOC} 2012, New York, NY,
  USA, May 19 - 22, 2012}, pages 19--26. {ACM},
\newblock 2012, \href {http://dx.doi.org/10.1145/2213977.2213980}
  {\path{doi:10.1145/2213977.2213980}}.

\bibitem[CG04]{CG04}
D.~E. Carroll and A.~Goel.
\newblock Lower bounds for embedding into distributions over excluded minor
  graph families.
\newblock In S.~Albers and T.~Radzik, editors, {\em Algorithms - {ESA} 2004,
  12th Annual European Symposium, Bergen, Norway, September 14-17, 2004,
  Proceedings}, volume 3221 of {\em Lecture Notes in Computer Science}, pages
  146--156. Springer,
\newblock 2004, \href {http://dx.doi.org/10.1007/978-3-540-30140-0\_15}
  {\path{doi:10.1007/978-3-540-30140-0\_15}}.

\bibitem[CGN{\etalchar{+}}06]{CGNRS06}
C.~Chekuri, A.~Gupta, I.~Newman, Y.~Rabinovich, and A.~Sinclair.
\newblock Embedding {$k$}-outerplanar graphs into {$\ell_1$}.
\newblock {\em SIAM J. Discrete Math.}, 20(1):119--136,
\newblock 2006, \href
  {http://dx.doi.org/http://dx.doi.org/10.1137/S0895480102417379}
  {\path{doi:http://dx.doi.org/10.1137/S0895480102417379}}.

\bibitem[CJLV08]{CJLV08}
A.~Chakrabarti, A.~Jaffe, J.~R. Lee, and J.~Vincent.
\newblock Embeddings of topological graphs: Lossy invariants, linearization,
  and 2-sums.
\newblock In {\em 49th Annual {IEEE} Symposium on Foundations of Computer
  Science, {FOCS} 2008, October 25-28, 2008, Philadelphia, PA, {USA}}, pages
  761--770,
\newblock 2008, \href {http://dx.doi.org/10.1109/FOCS.2008.79}
  {\path{doi:10.1109/FOCS.2008.79}}.

\bibitem[CLPP23]{CLPP23}
V.~Cohen{-}Addad, H.~Le, M.~Pilipczuk, and M.~Pilipczuk.
\newblock Planar and minor-free metrics embed into metrics of polylogarithmic
  treewidth with expected multiplicative distortion arbitrarily close to 1.
\newblock {\em CoRR}, abs/2304.07268, 2023.
\newblock
\newblock To appear in FOCS 2023, \href {http://arxiv.org/abs/2304.07268}
  {\path{arXiv:2304.07268}}, \href
  {http://dx.doi.org/10.48550/arXiv.2304.07268}
  {\path{doi:10.48550/arXiv.2304.07268}}.

\bibitem[CSW13]{CSW13}
C.~Chekuri, F.~B. Shepherd, and C.~Weibel.
\newblock Flow-cut gaps for integer and fractional multiflows.
\newblock {\em J. Comb. Theory, Ser. {B}}, 103(2):248--273,
\newblock 2013, \href {http://dx.doi.org/10.1016/j.jctb.2012.11.002}
  {\path{doi:10.1016/j.jctb.2012.11.002}}.

\bibitem[CW04]{CW04}
D.~Z. Chen and X.~Wu.
\newblock Efficient algorithms for k-terminal cuts on planar graphs.
\newblock {\em Algorithmica}, 38(2):299--316,
\newblock 2004, \href {http://dx.doi.org/10.1007/s00453-003-1061-2}
  {\path{doi:10.1007/s00453-003-1061-2}}.

\bibitem[CX00]{CX00}
D.~Z. Chen and J.~Xu.
\newblock Shortest path queries in planar graphs.
\newblock In {\em Proceedings of the Thirty-second Annual ACM Symposium on
  Theory of Computing}, STOC '00, pages 469--478, New York, NY, USA, 2000.
\newblock ACM, \href {http://dx.doi.org/10.1145/335305.335359}
  {\path{doi:10.1145/335305.335359}}.

\bibitem[EFN17]{EFN17}
M.~Elkin, A.~Filtser, and O.~Neiman.
\newblock Terminal embeddings.
\newblock {\em Theor. Comput. Sci.}, 697:1--36,
\newblock 2017, \href {http://dx.doi.org/10.1016/J.TCS.2017.06.021}
  {\path{doi:10.1016/J.TCS.2017.06.021}}.

\bibitem[EFN18]{EFN18}
M.~Elkin, A.~Filtser, and O.~Neiman.
\newblock Prioritized metric structures and embedding.
\newblock {\em {SIAM} J. Comput.}, 47(3):829--858,
\newblock 2018, \href {http://dx.doi.org/10.1137/17M1118749}
  {\path{doi:10.1137/17M1118749}}.

\bibitem[EGK{\etalchar{+}}14]{EGKRTT14}
M.~Englert, A.~Gupta, R.~Krauthgamer, H.~R{\"{a}}cke, I.~Talgam{-}Cohen, and
  K.~Talwar.
\newblock Vertex sparsifiers: New results from old techniques.
\newblock {\em {SIAM} J. Comput.}, 43(4):1239--1262,
\newblock 2014, \href {http://dx.doi.org/10.1137/130908440}
  {\path{doi:10.1137/130908440}}.

\bibitem[EMJ87]{EMV87}
R.~E. Erickson, C.~L. Monma, and A.~F.~V. Jr.
\newblock Send-and-split method for minimum-concave-cost network flows.
\newblock {\em Math. Oper. Res.}, 12(4):634--664,
\newblock 1987, \href {http://dx.doi.org/10.1287/moor.12.4.634}
  {\path{doi:10.1287/moor.12.4.634}}.

\bibitem[EN22]{EN22}
M.~Elkin and O.~Neiman.
\newblock Lossless prioritized embeddings.
\newblock {\em {SIAM} J. Discret. Math.}, 36(3):1529--1550,
\newblock 2022, \href {http://dx.doi.org/10.1137/21M1436221}
  {\path{doi:10.1137/21M1436221}}.

\bibitem[FFI{\etalchar{+}}24]{FFIKLMZ24}
A.~Filtser, T.~Friedrich, D.~Issac, N.~Kumar, H.~Le, N.~Mallek, and Z.~Zeif.
\newblock Optimal padded decomposition for bounded treewidth graphs.
\newblock {\em CoRR}, abs/2407.12230,
\newblock 2024, \href {http://dx.doi.org/10.48550/arXiv.2407.12230}
  {\path{doi:10.48550/arXiv.2407.12230}}.

\bibitem[FGK24]{FGK24}
A.~Filtser, L.~Gottlieb, and R.~Krauthgamer.
\newblock Labelings vs. embeddings: On distributed and prioritized
  representations of distances.
\newblock {\em Discret. Comput. Geom.}, 71(3):849--871,
\newblock 2024, \href {http://dx.doi.org/10.1007/S00454-023-00565-2}
  {\path{doi:10.1007/S00454-023-00565-2}}.

\bibitem[Fil19]{Fil19padded}
A.~Filtser.
\newblock On strong diameter padded decompositions.
\newblock In {\em Approximation, Randomization, and Combinatorial Optimization.
  Algorithms and Techniques, {APPROX/RANDOM} 2019, September 20-22, 2019,
  Massachusetts Institute of Technology, Cambridge, MA, {USA}}, pages
  6:1--6:21,
\newblock 2019, \href {http://dx.doi.org/10.4230/LIPIcs.APPROX-RANDOM.2019.6}
  {\path{doi:10.4230/LIPIcs.APPROX-RANDOM.2019.6}}.

\bibitem[Fil20]{Fil20}
A.~Filtser.
\newblock A face cover perspective to $\ell_1$- embeddings of planar graphs.
\newblock In S.~Chawla, editor, {\em Proceedings of the 2020 {ACM-SIAM}
  Symposium on Discrete Algorithms, {SODA} 2020, Salt Lake City, UT, USA,
  January 5-8, 2020}, pages 1945--1954. {SIAM},
\newblock 2020, \href {http://dx.doi.org/10.1137/1.9781611975994.120}
  {\path{doi:10.1137/1.9781611975994.120}}.

\bibitem[Fil23]{Fil23}
A.~Filtser.
\newblock Labeled nearest neighbor search and metric spanners via locality
  sensitive orderings.
\newblock In E.~W. Chambers and J.~Gudmundsson, editors, {\em 39th
  International Symposium on Computational Geometry, SoCG 2023, June 12-15,
  2023, Dallas, Texas, {USA}}, volume 258 of {\em LIPIcs}, pages 33:1--33:18.
  Schloss Dagstuhl - Leibniz-Zentrum f{\"{u}}r Informatik,
\newblock 2023, \href {http://dx.doi.org/10.4230/LIPICS.SOCG.2023.33}
  {\path{doi:10.4230/LIPICS.SOCG.2023.33}}.

\bibitem[FKS19]{FKS19}
E.~Fox{-}Epstein, P.~N. Klein, and A.~Schild.
\newblock Embedding planar graphs into low-treewidth graphs with applications
  to efficient approximation schemes for metric problems.
\newblock In T.~M. Chan, editor, {\em Proceedings of the Thirtieth Annual
  {ACM-SIAM} Symposium on Discrete Algorithms, {SODA} 2019, San Diego,
  California, USA, January 6-9, 2019}, pages 1069--1088. {SIAM},
\newblock 2019, \href {http://dx.doi.org/10.1137/1.9781611975482.66}
  {\path{doi:10.1137/1.9781611975482.66}}.

\bibitem[FL21]{FL21}
A.~Filtser and H.~Le.
\newblock Clan embeddings into trees, and low treewidth graphs.
\newblock In S.~Khuller and V.~V. Williams, editors, {\em {STOC} '21: 53rd
  Annual {ACM} {SIGACT} Symposium on Theory of Computing, Virtual Event, Italy,
  June 21-25, 2021}, pages 342--355. {ACM},
\newblock 2021, \href {http://dx.doi.org/10.1145/3406325.3451043}
  {\path{doi:10.1145/3406325.3451043}}.

\bibitem[FL22]{FL22}
A.~Filtser and H.~Le.
\newblock Low treewidth embeddings of planar and minor-free metrics.
\newblock In {\em 63rd {IEEE} Annual Symposium on Foundations of Computer
  Science, {FOCS} 2022, Denver, CO, USA, October 31 - November 3, 2022}, pages
  1081--1092. {IEEE},
\newblock 2022, \href {http://dx.doi.org/10.1109/FOCS54457.2022.00105}
  {\path{doi:10.1109/FOCS54457.2022.00105}}.

\bibitem[Fre91]{Fre91}
G.~N. Frederickson.
\newblock Planar graph decomposition and all pairs shortest paths.
\newblock {\em J. {ACM}}, 38(1):162--204,
\newblock 1991, \href {http://dx.doi.org/10.1145/102782.102788}
  {\path{doi:10.1145/102782.102788}}.

\bibitem[Fre95]{Fre95}
G.~N. Frederickson.
\newblock Using cellular graph embeddings in solving all pairs shortest paths
  problems.
\newblock {\em J. Algorithms}, 19(1):45--85,
\newblock 1995, \href {http://dx.doi.org/10.1006/jagm.1995.1027}
  {\path{doi:10.1006/jagm.1995.1027}}.

\bibitem[FRT04]{FRT04}
J.~Fakcharoenphol, S.~Rao, and K.~Talwar.
\newblock A tight bound on approximating arbitrary metrics by tree metrics.
\newblock {\em J. Comput. Syst. Sci.}, 69(3):485--497,
\newblock November 2004, \href {http://dx.doi.org/10.1016/j.jcss.2004.04.011}
  {\path{doi:10.1016/j.jcss.2004.04.011}}.

\bibitem[FT03]{FT03}
J.~Fakcharoenphol and K.~Talwar.
\newblock An improved decomposition theorem for graphs excluding a fixed minor.
\newblock In {\em RANDOM-APPROX}, pages 36--46,
\newblock 2003.

\bibitem[GKL03]{GKL03}
A.~Gupta, R.~Krauthgamer, and J.~R. Lee.
\newblock Bounded geometries, fractals, and low-distortion embeddings.
\newblock In {\em FOCS '03: Proceedings of the 44th Annual IEEE Symposium on
  Foundations of Computer Science}, page 534, Washington, DC, USA, 2003.
\newblock IEEE Computer Society.

\bibitem[GNRS04]{GNRS04}
A.~Gupta, I.~Newman, Y.~Rabinovich, and A.~Sinclair.
\newblock Cuts, trees and $\ell_1$s-embeddings of graphs.
\newblock {\em Combinatorica}, 24(2):233--269,
\newblock 2004, \href {http://dx.doi.org/10.1007/s00493-004-0015-x}
  {\path{doi:10.1007/s00493-004-0015-x}}.

\bibitem[HST86]{HST86}
C.~A.~J. Hurkens, A.~Schrijver, and {\'E}.~Tardos.
\newblock On fractional multicommodity flows and distance functions.
\newblock {\em Discrete Mathematics}, 73:99--109,
\newblock 1986, \href {http://dx.doi.org/10.1016/0012-365X(88)90137-9}
  {\path{doi:10.1016/0012-365X(88)90137-9}}.

\bibitem[KLMN05]{KLMN04}
R.~Krauthgamer, J.~R. Lee, M.~Mendel, and A.~Naor.
\newblock Measured descent: a new embedding method for finite metrics.
\newblock {\em Geometric and Functional Analysis}, 15(4):839--858,
\newblock 2005, \href {http://dx.doi.org/10.1007/s00039-005-0527-6}
  {\path{doi:10.1007/s00039-005-0527-6}}.

\bibitem[KLR19]{KLR19}
R.~Krauthgamer, J.~R. Lee, and H.~Rika.
\newblock Flow-cut gaps and face covers in planar graphs.
\newblock In {\em Proceedings of the Thirtieth Annual {ACM-SIAM} Symposium on
  Discrete Algorithms, {SODA} 2019, San Diego, California, USA, January 6-9,
  2019}, pages 525--534,
\newblock 2019, \href {http://dx.doi.org/10.1137/1.9781611975482.33}
  {\path{doi:10.1137/1.9781611975482.33}}.

\bibitem[KNvL20]{KNvL20}
S.~Kisfaludi{-}Bak, J.~Nederlof, and E.~J. van Leeuwen.
\newblock Nearly {ETH}-tight algorithms for planar steiner tree with terminals
  on few faces.
\newblock {\em {ACM} Trans. Algorithms}, 16(3):28:1--28:30,
\newblock 2020, \href {http://dx.doi.org/10.1145/3371389}
  {\path{doi:10.1145/3371389}}.

\bibitem[KPR93]{KPR93}
P.~Klein, S.~A. Plotkin, and S.~Rao.
\newblock Excluded minors, network decomposition, and multicommodity flow.
\newblock In {\em Proceedings of the twenty-fifth annual ACM symposium on
  Theory of computing}, STOC '93, pages 682--690, New York, NY, USA, 1993.
\newblock ACM, \href
  {http://dx.doi.org/http://doi.acm.org/10.1145/167088.167261}
  {\path{doi:http://doi.acm.org/10.1145/167088.167261}}.

\bibitem[KPZ19]{KPZ19}
N.~Karpov, M.~Pilipczuk, and A.~Zych{-}Pawlewicz.
\newblock An exponential lower bound for cut sparsifiers in planar graphs.
\newblock {\em Algorithmica}, 81(10):4029--4042,
\newblock 2019, \href {http://dx.doi.org/10.1007/s00453-018-0504-8}
  {\path{doi:10.1007/s00453-018-0504-8}}.

\bibitem[KR20]{KR20}
R.~Krauthgamer and H.~Rika.
\newblock Refined vertex sparsifiers of planar graphs.
\newblock {\em {SIAM} J. Discret. Math.}, 34(1):101--129,
\newblock 2020, \href {http://dx.doi.org/10.1137/17M1151225}
  {\path{doi:10.1137/17M1151225}}.

\bibitem[Kum22]{Kumar22}
N.~Kumar.
\newblock An approximate generalization of the okamura-seymour theorem.
\newblock In {\em 63rd {IEEE} Annual Symposium on Foundations of Computer
  Science, {FOCS} 2022, Denver, CO, USA, October 31 - November 3, 2022}, pages
  1093--1101. {IEEE},
\newblock 2022, \href {http://dx.doi.org/10.1109/FOCS54457.2022.00106}
  {\path{doi:10.1109/FOCS54457.2022.00106}}.

\bibitem[LLR95]{LLR95}
N.~Linial, E.~London, and Y.~Rabinovich.
\newblock The geometry of graphs and some of its algorithmic applications.
\newblock {\em Combinatorica}, 15(2):215--245,
\newblock 1995, \href {http://dx.doi.org/10.1007/BF01200757}
  {\path{doi:10.1007/BF01200757}}.

\bibitem[LMM15]{LMM15}
J.~R. Lee, M.~Mendel, and M.~Moharrami.
\newblock A node-capacitated okamura--seymour theorem.
\newblock {\em Mathematical Programming}, 153(2):381--415,
\newblock 2015, \href {http://dx.doi.org/10.1007/s10107-014-0810-0}
  {\path{doi:10.1007/s10107-014-0810-0}}.

\bibitem[LN05]{LN05}
J.~R. Lee and A.~Naor.
\newblock Extending lipschitz functions via random metric partitions.
\newblock {\em Inventiones Mathematicae}, 160(1):59--95,
\newblock 2005, \href {http://dx.doi.org/10.1007/s00222-004-0400-5}
  {\path{doi:10.1007/s00222-004-0400-5}}.

\bibitem[LR10]{LR10}
J.~R. Lee and P.~Raghavendra.
\newblock Coarse differentiation and multi-flows in planar graphs.
\newblock {\em Discrete {\&} Computational Geometry}, 43(2):346--362,
\newblock 2010, \href {http://dx.doi.org/10.1007/s00454-009-9172-4}
  {\path{doi:10.1007/s00454-009-9172-4}}.

\bibitem[LS09]{LS09}
J.~R. Lee and A.~Sidiropoulos.
\newblock On the geometry of graphs with a forbidden minor.
\newblock In {\em Proceedings of the 41st Annual {ACM} Symposium on Theory of
  Computing, {STOC} 2009, Bethesda, MD, USA, May 31 - June 2, 2009}, pages
  245--254,
\newblock 2009, \href {http://dx.doi.org/10.1145/1536414.1536450}
  {\path{doi:10.1145/1536414.1536450}}.

\bibitem[LS10]{LS10}
J.~R. Lee and A.~Sidiropoulos.
\newblock Genus and the geometry of the cut graph.
\newblock In {\em Proceedings of the Twenty-First Annual {ACM-SIAM} Symposium
  on Discrete Algorithms, {SODA} 2010, Austin, Texas, USA, January 17-19,
  2010}, pages 193--201,
\newblock 2010, \href {http://dx.doi.org/10.1137/1.9781611973075.18}
  {\path{doi:10.1137/1.9781611973075.18}}.

\bibitem[LS13]{LS13}
J.~R. Lee and A.~Sidiropoulos.
\newblock Pathwidth, trees, and random embeddings.
\newblock {\em Combinatorica}, 33(3):349--374,
\newblock 2013, \href {http://dx.doi.org/10.1007/s00493-013-2685-8}
  {\path{doi:10.1007/s00493-013-2685-8}}.

\bibitem[Mat02]{matbook}
J.~Matou{\v{s}}ek.
\newblock {\em Lectures on discrete geometry}.
\newblock Springer-Verlag, New York,
\newblock 2002, \href {http://dx.doi.org/10.1007/978-1-4613-0039-7}
  {\path{doi:10.1007/978-1-4613-0039-7}}.

\bibitem[Mil86]{Mil86}
G.~L. Miller.
\newblock Finding small simple cycle separators for 2-connected planar graphs.
\newblock {\em J. Comput. Syst. Sci.}, 32(3):265--279,
\newblock 1986, \href {http://dx.doi.org/10.1016/0022-0000(86)90030-9}
  {\path{doi:10.1016/0022-0000(86)90030-9}}.

\bibitem[MNS85]{MNS85}
K.~Matsumoto, T.~Nishizeki, and N.~Saito.
\newblock An efficient algorithm for finding multicommodity flows in planar
  networks.
\newblock {\em {SIAM} J. Comput.}, 14(2):289--302,
\newblock 1985, \href {http://dx.doi.org/10.1137/0214023}
  {\path{doi:10.1137/0214023}}.

\bibitem[OS81]{OS81}
H.~Okamura and P.~Seymour.
\newblock Multicommodity flows in planar graphs.
\newblock {\em Journal of Combinatorial Theory, Series B}, 31(1):75 -- 81,
\newblock 1981, \href
  {http://dx.doi.org/http://dx.doi.org/10.1016/S0095-8956(81)80012-3}
  {\path{doi:http://dx.doi.org/10.1016/S0095-8956(81)80012-3}}.

\bibitem[Rao99]{Rao99}
S.~Rao.
\newblock Small distortion and volume preserving embeddings for planar and
  {E}uclidean metrics.
\newblock In {\em Proceedings of the Fifteenth Annual Symposium on
  Computational Geometry, Miami Beach, Florida, USA, June 13-16, 1999}, pages
  300--306,
\newblock 1999, \href {http://dx.doi.org/10.1145/304893.304983}
  {\path{doi:10.1145/304893.304983}}.

\bibitem[Tho04]{T04}
M.~Thorup.
\newblock Compact oracles for reachability and approximate distances in planar
  digraphs.
\newblock {\em J. ACM}, 51(6):993--1024,
\newblock November 2004, \href {http://dx.doi.org/10.1145/1039488.1039493}
  {\path{doi:10.1145/1039488.1039493}}.

\end{thebibliography}
}


\end{document}